\documentclass[12pt]{article}

\usepackage{amssymb,latexsym,amsmath,amsthm,bm,amsfonts}
\usepackage{color,xcolor,setspace,multirow,booktabs}
\usepackage[round, authoryear]{natbib}
\usepackage[titletoc,title]{appendix}
\usepackage[top=1in, bottom=1in, left=0.7in, right=0.7in]{geometry}
\usepackage[utf8]{inputenc}
\usepackage{graphicx, float}
\usepackage{stmaryrd, textcomp, enumitem}

\definecolor{darkblue}{rgb}{0.0, 0.0, 0.4}
\usepackage[colorlinks=true,allcolors=darkblue]{hyperref}

\theoremstyle{plain}
\newtheorem{theorem}{Theorem}[section]

\newtheorem{lemma}{Lemma}[section]

\newlist{assumpenum}{enumerate}{1}
\setlist[assumpenum]{leftmargin=3.5em, label=C\arabic*), ref=C\arabic*, font=\bfseries}
\newlist{contrienum}{enumerate}{1}
\setlist[contrienum]{leftmargin=3.5em, label=\arabic*), ref=\arabic*, font=\bfseries}
\newlist{lemmaitemenum}{enumerate}{1}
\setlist[lemmaitemenum]{leftmargin=3.5em, label=(\thelemma.\arabic*), ref=(\thelemma.\arabic*)}

\def\bs#1{\boldsymbol{#1}}

\newcommand\independent{\protect\mathpalette{\protect\independenT}{\perp}}
\def\independenT#1#2{\mathrel{\rlap{$#1#2$}\mkern2mu{#1#2}}}

\usepackage{authblk}

\usepackage[symbol]{footmisc}

\begin{document}
\allowdisplaybreaks
\doublespacing

\baselineskip=18pt

\title{{\bf Improved Efficiency for Cross-Arm Comparisons via Platform Designs}}
\author[1]{Tzu-Jung Huang}
\author[1,2]{Alex Luedtke}
\author[ \hspace{-1ex}]{the AMP Investigators Group\footnote{Complete list of investigators in the supplement.}}

\affil[1]{\small Department of Statistics, University of Washington}
\affil[2]{\small Vaccine and Infectious Disease Division, Fred Hutchinson Cancer Research Center}

\date{\today}
\maketitle

\fontsize{12}{16pt plus.8pt minus .6pt}\selectfont

\begin{abstract}
  Though platform trials have been touted for their flexibility and streamlined use of trial resources, their statistical efficiency is not well understood. 
  We fill this gap by establishing their greater efficiency for comparing the relative efficacy of multiple interventions over using several separate, two-arm trials, 
  where the relative efficacy of an arbitrary pair of interventions is evaluated by contrasting their relative risks as compared to control. In theoretical and numerical studies, we demonstrate that the inference of such a contrast using data from a platform trial enjoys identical or better precision than using data from separate trials, even when the former enrolls substantially fewer participants. This benefit is attributed to the sharing of  
  controls among interventions under contemporaneous randomization, which is a key feature of platform trials. 
  We further provide a novel procedure for establishing the non-inferiority of a given intervention relative to the most efficacious of the other interventions under evaluation, where this procedure is adaptive in the sense that it need not be \textit{a priori} known which of these other interventions is most efficacious. 
  Our numerical studies show that this testing procedure can attain substantially better power when the data arise from a platform trial rather than multiple separate trials. Our results are illustrated using data from two monoclonal antibody trials for the prevention of HIV.
\end{abstract}

\section{Introduction}
This work is motivated by the World Health Organization's (WHO's) Solidarity Trials for Covid-19 vaccines \citep{Krause2020} and treatments \citep{WHOSolidarityTrialConsortium2021}, which seek to concurrently evaluate multiple candidate vaccines and treatments to contain the burden and spread of Covid-19. To achieve this aim, the investigators are using a platform design \citep{Sridhara2015, Saville2016, Woodcock2017}. Such designs make it possible to complete a simultaneous evaluation of various public health interventions within a single trial. In contrast to more traditional multi-arm designs, platform trials allow for the evaluation of candidate interventions that may be available at different times and various international sites.
Having this flexibility increases the chances of generating reliable evidence to determine which interventions work effectively.
Owing to these advantages, platform designs have recently been advocated to evaluate candidate interventions during disease outbreaks \citep{Dean2020}.

Platform trials make more efficient use of resources than do separate, independently-conducted trials by reducing the number of participants enrolled on the control arm \citep{Woodcock2017}. 
Consequently, trial resources can be redistributed to enroll participants on the active intervention arms. For sponsors, this redistribution has the benefit of potentially decreasing the total required sample size, while, for participants, it has the appeal of increasing the chance that they will receive an experimental intervention that they could not receive outside of the trial. Moreover, by using centralized governance, platform trials amortize the cost of establishing study sites across multiple active interventions and can ensure common eligibility criteria and study procedures across all intervention-control comparisons. In addition, as each candidate active intervention is monitored for early evidence of benefit or harm, employing a platform design makes it possible to target trial resources to the study of active interventions that are more likely to be successful \citep{Saville2016}. 

Methods for analyzing platform trial data should be able to accommodate key features of the design, namely that candidate active interventions can be added to the trial once they become available, and those that show a lack of efficacy can be eliminated.
To this end, several authors have advocated that, for each active intervention, controls that are under contemporaneous randomization should be used as comparators \citep{Lee2020, Kopp-Schneider2020, Lee2021}.
Methods that restrict to contemporaneous comparisons can avoid estimation bias and anticonservative confidence interval coverage that may arise due to a potential temporal trend in the outcome or baseline characteristics of the control arm.
Moreover, during periods when multiple active interventions are under randomization simultaneously, the contemporaneous control comparators may be shared between these various arms.

Analysis methods should also be able to accommodate the fact that, once an active intervention has been found to be efficacious, it may serve as an active control in the evaluation of the other candidate active interventions \citep{Woodcock2017, Kaizer2018}. One approach to make this evaluation would involve directly comparing the outcomes on a candidate active intervention to those on the active control. Unfortunately, this approach would be susceptible to the same temporal-trend-induced biases as described above. A less direct, but more robust, method involves comparing the two interventions by contrasting their efficacies relative to their respective contemporaneous control arms. Under a constancy assumption that underlies the validity of many methodologies developed for non-inferiority trials \citep{everson2010bio}, this strategy will yield unbiased comparisons of the efficacies of the two interventions.

Though standard analyses can be used to evaluate the efficacy of each active intervention compared to its contemporaneous control, more care is needed to quantify uncertainty when contrasting the efficacies of various active interventions.
The main challenge is that the overlap in the contemporaneous control groups for different interventions induces a positive correlation between their corresponding efficacy estimates. 
This positive correlation has been used as a means to justify not doing multiplicity adjustments in multi-arm trials \citep{howard2018recommendations}. In particular, the positive correlation induces a reduction in the family-wise error rate compared to independent separate trials, suggesting that multiplicity adjustment should not be required solely due to sharing control data. 
Since multiple testing corrections would not have been employed had separate trials been conducted, platform trials \citep{WHOSolidarityTrialConsortium2021,howard2021platform} have often followed multi-armed trials \citep{freidlin2008multi} in not including them. 
In this work, we follow this precedent and do not consider multiple testing in the platform trial setting. 
In another line of research, recent works have provided Bayesian strategies to describe and accommodate this positive correlation in platform trials when the outcomes are binary \citep{Saville2016,Hobbs2018,Kaizer2018}.
To the best of our knowledge, there is no available frequentist method that properly addresses this issue, nor is there any approach that is applicable when the outcomes may be right censored.
In this work, we introduce a theoretically-grounded framework to account for this correlation when making cross-intervention comparisons with time-to-event outcomes based on platform trial data.

Our contributions are as follows:
\begin{contrienum}
  \item \label{contribution.1} 
  In Section~\ref{sec:inf_asymp_RR}, we establish the joint asymptotic normality of intervention-specific conditional relative risk estimators. The limiting covariance matrix quantifies the positive correlation induced from sharing control participants between the active interventions. 
  
  \item \label{contribution.2} In Section~\ref{sec:efficiency_gain}, we show that natural estimators of the relative efficacy of two interventions are statistically more efficient when applied to platform trial data rather than on data from separate, intervention-specific trials.
  This enables a substantial reduction in the sample size of the platform trials relative to that of the pooled separate trials, without sacrificing precision for evaluating efficacy.
  
  \item \label{contribution.3} In Section~\ref{sec:intersection_test}, we develop a noninferiority test that evaluates the efficacy of a given intervention relative to the most efficacious of the remaining interventions. This test is adaptive in the sense that there is no need to know, in advance, which of the remaining interventions should serve as the benchmark.
\end{contrienum}
The randomization scheme, data structure, and assumptions 
that shape the framework of our proposed methods are introduced in Section~\ref{sec:preliminary}.
Numerical studies that support our theoretical findings are presented in Section~\ref{sec:numerical}, and an illustration of the proposed methods on data from the Antibody Mediated Prevention (AMP) trials \citep{corey2021two} can be found in Section~\ref{sec:data_analysis}. Section~\ref{sec:conclusion} closes with some concluding remarks.


\section{Preliminaries} \label{sec:preliminary}
\subsection{Randomization scheme}
We consider platform designs in which the availability of active interventions may vary over discrete windows defined by time and location. 
Within each location, windows are contiguous and may be of unequal widths: a window starts when an active intervention is placed under or removed from randomization and ends when the next window begins. 
To simplify presentation, we focus on a particular randomization scheme where, within each window, equal numbers of participants are expected to be enrolled to each available active intervention arm and to the control arm. More concretely, within an arbitrary window in which $k$ active interventions are under randomization, each active intervention is assigned with probability $1/(k+1)$, and control is also assigned with probability $1/(k+1)$. 
In settings where there are intervention-specific matched controls, a participant assigned to control is further randomized to one of the $k$ matched controls, so that each matched control is assigned with probability $1/[k(k+1)]$. 

An example of this randomization scheme is illustrated in Table \ref{tab:randomization scheme}, which represents a modification of a figure from the protocol for the WHO Solidarity Trial for Vaccines \citep{WHOR&DBlueprint2020}. 
Three candidate active interventions are considered in this example, labeled as ${\rm A}_1$, ${\rm A}_2$, and ${\rm A}_3$, and corresponding matched controls are under concurrent randomization. Intervention ${\rm A}_1$ and its matched control ${\rm C}_1$ are under randomization in Windows 1-3 and 5; intervention ${\rm A}_2$ and its matched control ${\rm C}_2$ are under randomization in Windows 2-4, and intervention ${\rm A}_3$ and its matched control ${\rm C}_3$ are under randomization in Windows 3-5. 
When assessing the efficacy of a given active intervention, the control group used for comparison consists of all participants who were enrolled to a matched or unmatched control arm during a period in which the active intervention was under randomization.
This union of the matched and unmatched controls is referred to as the shared control arm. 
For example, for intervention ${\rm A}_2$, the shared control arm consists of $({\rm C}_1,{\rm C}_2)$ for window 2, $({\rm C}_1,{\rm C}_2,{\rm C}_3)$ for window 3, and $({\rm C}_2,{\rm C}_3)$ for window 4.


\begin{table}[ht]
\centering
\begin{tabular}[t]{p{3.2cm}lllll}
  \toprule
    & Window 1  &  Window 2  &  Window 3 & Window 4 
    & Window 5\\
  \cmidrule{2-6}     
    \multirow{4}{*}{\parbox{4.5cm}{Intervention :\\
    Matched Control}} & 
    ${\rm A}_1:{\rm C}_1$ & ${\rm A}_1{\rm A}_1:{\rm C}_1$ & ${\rm A}_1{\rm A}_1{\rm A}_1:{\rm C}_1$ &  & ${\rm A}_1{\rm A}_1:{\rm C}_1$ \\
    &   & ${\rm A}_2{\rm A}_2:{\rm C}_2$ & ${\rm A}_2{\rm A}_2{\rm A}_2:{\rm C}_2$ & ${\rm A}_2{\rm A}_2:{\rm C}_2$ & \\
    &   &     & ${\rm A}_3{\rm A}_3{\rm A}_3:{\rm C}_3$ & ${\rm A}_3{\rm A}_3:{\rm C}_3$ & ${\rm A}_3{\rm A}_3:{\rm C}_3$ \\[.5em]
      & (1:1) & (2:1) & (3:1) & (2:1) & (2:1) \\
  \midrule
    \multirow{4}{*}{\parbox{4.5cm}{Intervention :\\
    Shared Control}} & 
    ${\rm A}_1:{\rm C}_1$ & ${\rm A}_1{\rm A}_1:\hspace{0.35em}\rotatebox[origin=b]{270}{\large$\Rsh$}$ & ${\rm A}_1{\rm A}_1{\rm A}_1:\hspace{0.7em}\rotatebox[origin=b]{270}{\large$\Rsh$}$ &  & ${\rm A}_1{\rm A}_1:\hspace{0.35em}\rotatebox[origin=b]{270}{\large$\Rsh$}$ \\
    &   & ${\rm A}_2{\rm A}_2:{\rm C}_1{\rm C}_2$ & ${\rm A}_2{\rm A}_2{\rm A}_2:{\rm C}_1{\rm C}_2{\rm C}_3$ & ${\rm A}_2{\rm A}_2:{\rm C}_2{\rm C}_3$ & $\phantom{{\rm A}_2{\rm A}_2:\;}{\rm C}_1{\rm C}_3$ \\
    &   &     & ${\rm A}_3{\rm A}_3{\rm A}_3:\hspace{0.7em}\rotatebox[origin=c]{270}{\large$\Lsh$}$ & ${\rm A}_3{\rm A}_3:\hspace{0.35em}\rotatebox[origin=c]{270}{\large$\Lsh$}$ & ${\rm A}_3{\rm A}_3:\hspace{0.35em}\rotatebox[origin=c]{270}{\large$\Lsh$}$ \\[0.5em]
      & (1:1) & (1:1) & (1:1) & (1:1) & (1:1) \\
  \bottomrule
\label{tab:randomization scheme}  
\end{tabular}
\caption{Allocation ratios of each active intervention versus its matched control and versus shared control in a platform trial.}
\end{table}

\subsection{Data structure, assumptions, and estimands} \label{sec:data_structure and assumptions}

We now describe the variables that are measured for each participant in the considered platform trial. The window in which a participant enrolls is denoted by $W$.
The baseline covariate to be used for stratification or adjustment is denoted by $Z$. We suppose that the covariate is discrete with finite support --- it is possible that $Z$ arises by combining multiple discrete covariates or by discretizing one or more continuous covariates into categorial subgroups. The randomization arm is indicated via a categorical variable $A$, which takes the value $0$ for control and the value $j$ for intervention $j$. To ease presentation, hereafter we denote the $k$ active interventions by $A=1$, $A=2$, etc., rather than by ${\rm A}_1$, ${\rm A}_2$, etc., as was done in Table \ref{tab:randomization scheme}.
The observed time $X$ is defined as the minimum of a continuous event time $T$ and a censoring time $C$, and $\Delta = 1(T \leq C)$ is a corresponding indicator of having observed the event. Enrollment is treated as time zero, so that $T=0$ indicates that the participant experienced the event exactly at the time of enrollment.  Interventions and windows are labeled sequentially, so that the first active intervention (window) is labeled ``intervention (window) 1'', the second is labeled ``intervention (window) 2'', and so on. We use $k$ to denote the total number of active interventions, $[k]=\{1,\ldots,k\}$ to denote the set of all active interventions, and $q$ to denote the number of windows over the course of the trial. For an active intervention $a$, we let $\mathbf{w}_a$ denote the set containing the windows in which intervention $a$ is under randomization. We observe $n$ independent and identically distributed (iid) copies $U_i=(X_i, \Delta_i, A_i, W_i, Z_i)$, $i=1,\ldots,n$, drawn from a distribution $P$.

We will make use of the following condition, which often holds in randomized trial settings:
\begin{assumpenum}[series=assumptions]
  \item \label{assump:indep_AZ_on_W} \textit{Randomized arm assignment:} The baseline covariate $Z$ is independent of $A$ conditional on $W$.
\end{assumpenum}
We will also make use of a constancy condition that the efficacy of an active intervention, defined in terms of a relative risk, should be stable across windows within specified strata.
These strata are defined via the value of a coarsening $V$ of $Z$. More concretely, $V=g(Z)$ for a known, possibly many-to-one, function $g$; we similarly let $V_i=g(Z_i)$. 
By taking $g$ to be the identity or a constant function, we could make $V$ correspond to the original baseline covariate $Z$ or a degenerate random variable, respectively. Alternatively, if $Z$ results from a combination of several discrete covariates (such as age group and sex), then $g$ could be chosen as a coordinate projection that returns a particular one of those covariates (such as age group). Henceforth, we will take $v$ to be a generic realization of $V$ corresponding to a generic realization $z$ of $Z$, that is, $v=g(z)$. 

As the event rate may vary across windows, so does the survival function $t\mapsto P(T>t|A=a, W=w, V=v)$ conditionally on being randomized to intervention $a$, enrolled in window $w$, and belonging to the stratum $v$. The corresponding conditional relative risk is
\begin{align}
    \gamma(t|a,w,v):=\frac{P(T\le t|A=a, W=w, V=v)}{P(T\le t|A=0, W=w, V=v)}. \label{eq:condRR}
\end{align}
To avoid dividing by zero or considering degenerate cases where it is all but obvious that an intervention is efficacious or not, we suppose throughout that $P(T\le t|A=a, W=w, V=v)>0$ for all possible values of $a$, namely when $a$ is control or any active intervention. 
The constancy condition at a specified time $t$ can be stated as follows.
\begin{assumpenum}[resume=assumptions]
  \item \label{assump:constancy} \textit{Constancy condition at time $t$:} For each active intervention $a$ and stratum $v$, there exists a $\gamma(t|a,v)\in (0,\infty)$ such that, for all windows $w$ for which $P(A=a,V=v|W=w)>0$, $\gamma(t|a,w,v)=\gamma(t|a,v)$.
\end{assumpenum}
The above condition states that the conditional relative risks of the active interventions are invariant across windows, and is plausible in many settings, for example, in vaccine trials \citep[e.g.,][]{Fleming2021, Follmann2021,Tsiatis2021}. In fact, the plausibility of this condition has been systematically evaluated across many disease areas, due to its importance to noninferiority analyses \citep{DAgostino2003,Fleming2008,Mauri2017,Zhang2019,May2020}. In Figure \ref{fig:illustration_of_constancy_assumption}, we provide an illustration of the implications of \ref{assump:constancy} in a particular example. The above is distinct from proportional hazards assumptions that are often employed in analyses based on the Cox model. Indeed, unlike proportional hazards assumptions that require a constant hazard ratio between an intervention and the control, \ref{assump:constancy} allows for the efficacy of each active intervention, quantified in terms of relative risk, to vary over time since enrollment. This flexibility is important, for example, in vaccine studies, where vaccine efficacy is often low shortly after inoculation, ramps as the immune response builds and booster shots are administered, and subsequently wanes over time. 
The results in this work do not rely on a proportional hazards assumption.

Under \ref{assump:constancy} at time $t$, the relative risk estimand further writes as
\begin{align}\label{eq:gammatav}
  \gamma(t|a,v)&= \frac{E[P(T \le t|A=0, W, V)\gamma(t|a,v)|A=a, W \in \mathbf{w}_a, V=v]}{E[P(T \le t|A=0, W, V)|A=0, W \in \mathbf{w}_a, V=v]} \nonumber \\
  &= \frac{E[P(T \le t|A=a, W, V)|A=a, W \in \mathbf{w}_a, V=v]}{E[P(T \le t|A=0, W, V)|A=0, W \in \mathbf{w}_a, V=v]} = \frac{P(T \le t|A=a, W \in \mathbf{w}_a, V=v)}{P(T \le t|A=0, W \in \mathbf{w}_a, V=v)},
\end{align}
where the first equality simply multiplies the left-hand side by one, the second holds by \ref{assump:constancy}, and the third holds by the law of total expectation. 
Note that we require that $P(W \in \mathbf{w}_a, V=v)>0$ for all $a$ and $v$ of interest, which guarantees that each considered relative risk estimand is well-defined. 
In this work, we will have three primary interests regarding the relative risk estimands defined above: (i) estimating the relative risk $\gamma(t|a,v)$ for one or more given interventions $a$ and strata $v$ of $V$, (ii) contrasting $\gamma(t|a,v)$ and $\gamma(t|a',v)$ for two interventions $a$ and $a'$, and (iii) in a noninferiority analysis, contrasting $\gamma(t|a,v)$ against $\min_{a'\not=a}\gamma(t|a',v)$.

The constancy condition \ref{assump:constancy} only applies to active interventions $a$, covariate strata $v$, and windows $w$ for which $P(A=a,V=v|W=w)>0$. This requirement is natural given that, when it fails, the conditional relative risk in \eqref{eq:condRR} is not even well-defined. Indeed, if individuals from covariate stratum $v$ are never enrolled in window $w$ --- say, because individuals from this stratum do not live near the location to which window $w$ pertains --- then it is not even clear how to define the relative risk of an intervention $a$ among those individuals in stratum $v$ who enroll in window $w$. Even if we restrict attention to windows $w$ in which covariate stratum $v$ is enrolled with positive probability, a given intervention $a$ will likely not be under randomization in some of these windows, rendering the conditional relative risk in \eqref{eq:condRR} ill-defined. To define the effect of intervention $a$ in such windows, counterfactual reasoning can be applied \citep[e.g.,][]{hernan2020causal}. For each participant, this involves conceptualizing a counterfactual event time $T(a)$ that would have occurred if, possibly contrary to fact, intervention $a$ had been under randomization at their time of enrollment and they had been randomized to that intervention. Under standard causal assumptions (idem, pages 5-6), the counterfactual relative risk $\gamma^{\textnormal{c}}(t|a,w,v):=P\{T(a)\le t|W=w,V=v\}/P\{T(0)\le t|W=w,V=v\}$ is equal to the relative risk in \eqref{eq:condRR} whenever that quantity is well-defined, that is, whenever $P(A=a,V=v|W=w)>0$. Hence, under the natural counterfactual extension of the constancy condition that says that $\gamma^{\textnormal{c}}(t|a,w,v)=\gamma(t|a,v)$ for all $w$ that are such that $P(V=v|W=w)>0$, the relative risk estimand of interest can be interpreted as defining the effect of intervention $a$ in all windows in which participants with $V=v$ are under enrollment, rather than just in those windows in which participants with $V=v$ are under enrollment and $a$ is under randomization. We leave further consideration of this counterfactual constancy assumption to future work, and instead focus in the remainder on making inference about the quantity $\gamma(t|a,v)$ defined in \ref{assump:constancy}, which is well-defined even in the absence of causal assumptions.

Though the relative risk of interest does not depend on the full collection of baseline covariates $Z$, the conditionally independent censoring condition \ref{assump:censoring} below generally makes it necessary to make use of $Z$ when estimating the relative risk, where $\independent$ is used to denote (conditional) independence:
\begin{assumpenum}[resume=assumptions]
  \item \label{assump:censoring} \textit{Conditionally independent censoring:} For each active intervention $a\in [k]$ and covariate stratum $z$ with $P(Z=z|A=a)>0$, it holds that $C\independent T\mid \{A=a,Z=z\}$ and $C\independent T\mid \{A=0,W\in\mathbf{w}_a,Z=z\}$.
\end{assumpenum}
This condition will often be more plausible than the usual independent censoring assumption made in clinical trial analyses, which states that $C$ is independent of $T$ conditionally on $A$. Indeed, the random variable $Z$ can contain baseline factors that are predictive of both the event and censoring times. For example, in an infectious disease outbreak setting, these factors may include behavioral risk information and the calendar time of enrollment.
The key insight is that, under \ref{assump:censoring}, the conditional survival function $S(t|a,\mathbf{w},z):= P(T>t|A=a, W \in \mathbf{w}, Z=z)$ is identifiable as a functional of the distribution of the observed data \citep{Dabrowska1989}.
Consequently, to estimate $\gamma(t|a,v)$, it suffices to write this quantity as a function of this conditional survival function and the distribution of $(W,Z)$. Lemma~\ref{lem:covAdjEstimand} in the appendix shows that this is indeed possible. In particular, under the randomization of arm assignment \ref{assump:indep_AZ_on_W}, the numerator and denominator on the right-hand side of \eqref{eq:gammatav} can be rewritten to show that
\begin{align} \label{eq:estimand}
  \gamma(t|a,v) &= \frac{1-\sum_{z \in \mathcal{Z}}S(t|a, \mathbf{w}_a, z)\,P(Z=z|W \in \mathbf{w}_a, V=v)}{1-\sum_{z \in \mathcal{Z}}S(t|0,\mathbf{w}_a,z)\,P(Z=z|W \in \mathbf{w}_a, V=v)}.
\end{align}
Even if \ref{assump:constancy} were to fail to hold, $\gamma(t|a,v)$ can be defined as above. In this case, $\gamma(t|a,v)$ represents particular summary of the efficacy of intervention $a$ across windows. Nevertheless, \ref{assump:constancy} is advantageous since, when it holds, contrasts of the relative risks of different interventions can be interpreted as being independent of the particular windows under consideration. 
In Section~\ref{sec:estimation_RR} we will present so-called plug-in estimators that estimate each survival function and conditional probability above, and then insert them into the right-hand side of \eqref{eq:estimand} to estimate the relative risk of interest.

We conclude by introducing some other notation that will appear throughout in this article. 
We denote a $k$-variate column vector $(b_1,\ldots,b_k)$ by $(b_{a})_{a \in [k]}$ and a set $\{b_1,\ldots,b_k\}$ by $\{b_a\}_{a\in [k]}$. Let $\mathbf{W} = \{\mathbf{w}_a\}_{a\in [k]}$. We also let $\mathbf{w}_0$ denote the collection of windows where the control is under randomization --- by design, $\mathbf{w}_0=\cup_{a\in [k]} \mathbf{w}_a$. Moreover, let $\tau$ be the end of follow-up relative to enrollment, $\mathcal{T} = [0, \tau]$ denote the range of possible observed times, $\mathcal{A} = \{0, 1, \ldots, k\}$ denote the set of all arms, and $\mathcal{Z}$ and $\mathcal{V}$ be the support of $Z$ and $V$, respectively. We will write $E$ to denote the expectation operator under $P$. We write $u=(x, \delta, a, w, z)$ to denote a generic realization of $U$.

\begin{figure}[hbt!]
\begin{center}
 \includegraphics[trim={0.1cm 0.1cm 0.2cm 0cm}, clip,  width=0.6\textwidth]{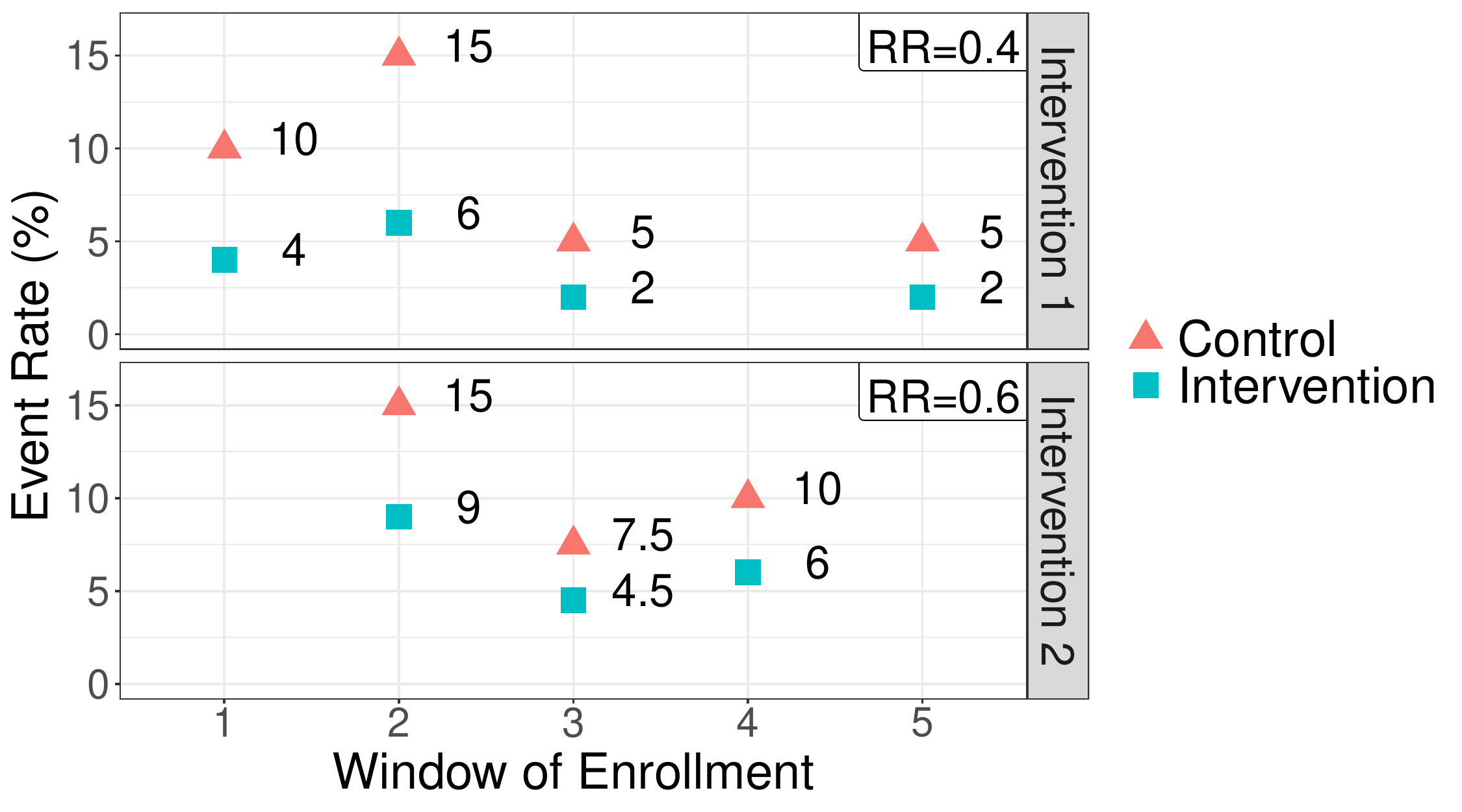}
 \caption{The relative risk (RR) of interventions 1 and 2 stays constant at 0.4 and 0.6, respectively, even though the event rates after randomization vary over windows of enrollment.}
 \label{fig:illustration_of_constancy_assumption}
\end{center}
\end{figure}

\section{Inference and asymptotic properties} \label{sec:inf_asymp_RR}
\subsection{Estimated conditional relative risk} \label{sec:estimation_RR}
Our estimators of the conditional survival function in \eqref{eq:estimand} will be based on the identity $S(t|a,\mathbf{w},z) = \exp[-\Lambda(t|a,\mathbf{w},z)]$, where $\Lambda(t|a,\mathbf{w},z)$ denotes the cumulative hazard function at time $t$ conditional on a randomization arm $a\in\mathcal{A}$, enrollment windows belonging to a specified window set $\mathbf{w}\in\mathbf{W}$, 
and a baseline covariate stratum $z$ such that $P(A=a,W\in\mathbf{w},Z=z)>0$. Specifically, this function is defined as
\begin{align*}
    \Lambda(t|a,\mathbf{w},z):= \int_0^t \frac{-S(dt'|a,\mathbf{w},z)}{S(t'|a,\mathbf{w},z)}.
\end{align*}
Henceforth we will refer to $\Lambda(\,\cdot\,|a,\mathbf{w},z)$ as the conditional CHF.

Though we will be most interested in estimating $\Lambda(\,\cdot\,|a,\mathbf{w}_a,z)$ and $\Lambda(\,\cdot\,|0,\mathbf{w}_a,z)$ for one or more given interventions $a$, it will simplify presentation to define an estimator of $\Lambda(t|a,\mathbf{w},z)$ for a generic arm $a$, set of windows $\mathbf{w}$, and covariate stratum $z$ for which $P(A=a,W\in\mathbf{w},Z=z)>0$. We will estimate this conditional CHF using the conditional Nelson-Aalen estimator. To introduce this estimator, we define the stratified basic counting process in the stratum $(a,\mathbf{w},z)$ and the average size of the risk set at time $t$ as $N_n(t;a,\mathbf{w},z) = n^{-1}\sum_{i=1}^n1(X_i \leq t, \Delta_i=1, A_i=a, W_i \in \mathbf{w}, Z_i=z)$ and $Y_n(t;a,\mathbf{w},z) = n^{-1}\sum_{i=1}^n1(X_i \geq t, A_i=a, W_i \in \mathbf{w}, Z_i=z)$, respectively. The conditional Nelson-Aalen estimator at time $t$ is given by $\Lambda_n(t|a,\mathbf{w},z) = \int_0^t [1/Y_n(s;a, \mathbf{w}, z)]N_n(ds;a, \mathbf{w}, z)$, and the conditional survival function can be estimated by $S_n(t|a,\mathbf{w},z) \equiv \exp[-\Lambda_n(t|a,\mathbf{w},z)]$.
For $\mathbf{w} \in \mathbf{W}$ and $v \in \mathcal{V}$, let $P(Z=z|W \in \mathbf{w}, V=v)$ be estimated by $P_n(z|\mathbf{w},v) = \sum_{i=1}^n1(Z_i=z, W_i \in \mathbf{w}, V_i=v)/\sum_{i=1}^n1(W_i \in \mathbf{w}, V_i=v)$. Inserting these estimators into \eqref{eq:estimand} naturally suggests a plug-in estimator of $\gamma(t|a,v)$:
\begin{align*}
  \gamma_n(t|a,v) = \frac{1-\sum_{z \in \mathcal{Z}}S_n(t|a,\mathbf{w}_a,z)P_n(z|\mathbf{w}_a,v)}{1-\sum_{z \in \mathcal{Z}}S_n(t|0,\mathbf{w}_a,z)P_n(z|\mathbf{w}_a,v)}.
\end{align*}

\subsection{Asymptotic normality} \label{sec:asymp_dist_RR}
To study the large-sample behavior of our proposed estimator of $\gamma(t|a,v)$,
it will be helpful to have characterized the joint asymptotic behavior of the estimators of the conditional survival functions $S(\,\cdot\,|0,\mathbf{w}_a,z)$ and $S(\,\cdot\,|a,\mathbf{w}_a,z)$ across different interventions $a$.
This characterization is simplified by the fact that participants on different randomization arms are mutually exclusive --- as a consequence of this fact, $S_n(t|a,\mathbf{w}_a,z)$ is independent of $S_n(t|0,\mathbf{w}_{a'},z)$ for every intervention $a'$ and of $S_n(t|a',\mathbf{w}_{a'},z)$ for every intervention arm $a' \not= a$. The joint distribution of the conditional survival functions for the shared control arms within window sets in $\mathbf{W}$ is more involved. 
Owing to the control-sharing between interventions $a$ and $a'$,
$S_n(t|0,\mathbf{w}_a,z)$ is not independent of $S_n(t|0,\mathbf{w}_{a'},z)$ unless $a$ and $a'$ are never under contemporaneous randomization --- that is, unless $\mathbf{w}_a\cap \mathbf{w}_{a'}=\emptyset$. Except in degenerate cases, the dependence is otherwise inevitable due to the mutual inclusion of participants who are under contemporaneous randomization in the estimation of $S(t|0,\mathbf{w}_a,z)$ and $S(t|0,\mathbf{w}_{a'},z)$.
To characterize this dependence, we derive the asymptotic behavior of the estimated stratified control-arm survival functions across 
all the intervention-specific window sets, namely
$(S_n(t|0,\mathbf{w}_{a},z))_{a \in [k]}$ --- see Lemma \ref{lemma:asymp_surv_functions} in the appendix for details. In the same lemma, we also establish the asymptotic normality of the estimated stratified intervention-arm survival functions across their designated window sets,
$(S_n(t|a,\mathbf{w}_{a},z))_{a \in [k]}$.

To derive the aforementioned distributional results, we establish that the estimators of the conditional survival functions are asymptotically linear. We recall that an estimator $\hat{\psi}$ of some estimand $\psi$ is called asymptotically linear if there exists a mean-zero, finite-variance function $f$, which typically depends on $P$, such that
\begin{align*}
    \hat{\psi} - \psi&= n^{-1}\sum_{i=1}^n f(U_i) + o_p(n^{-1/2}),
\end{align*}
where we write $o_p(r_n)$ to denote a term that converges to zero in probability as $n\rightarrow\infty$ even once divided by $r_n$. The function $f$ is known as the influence function of $\hat{\psi}$. The joint limiting distribution of several asymptotically linear estimators $\hat{\psi}_1,\ldots,\hat{\psi}_k$ can be derived using Slutsky's lemma and the central limit theorem: indeed, letting $f_1,\ldots,f_k$ denote the influence functions of $\hat{\psi}_1,\ldots,\hat{\psi}_k$, we have that $n^{1/2}(\hat{\psi}_a-\psi_a)_{a \in [k]} \overset{d}{\rightarrow} N(0,\Sigma)$, where 
$\Sigma={\rm Cov}_P[(f_a(U))_{a \in [k]}]$. A delta method is also available for asymptotically linear estimators, which makes it possible to compute the influence function of a real-valued function of one or more asymptotically linear estimators $\hat{\psi}_1,\ldots,\hat{\psi}_k$ via the dot product of the gradient of the function and the influence functions of $\hat{\psi}_1,\ldots,\hat{\psi}_k$.

We use such a delta-method argument to translate the asymptotic linearity of $S_n(t|a,\mathbf{w}_a,z)$, $S_n(t|0,\mathbf{w}_a,z)$, and $P_n(z|\mathbf{w}_a,v)$ over $z\in\mathcal{Z}$ into an asymptotic linearity result for the estimator $\varphi_n(t|a,v) := \log[\gamma_n(t|a,v)]$ of the log-relative risk $\varphi(t|a,v) := \log[\gamma(t|a,v)]$. 
Before presenting the result, we define some needed notation. Let $p_1(t,a,\mathbf{w},z):=P(X \leq t, \Delta=1, A=a, W \in \mathbf{w}, Z=z)$, $p_2(t,a,\mathbf{w},z):=P(X \geq t, A=a, W \in \mathbf{w}, Z=z)$, and 
\begin{align} \label{eq:notations_probs}
  &p(z|\mathbf{w},v) := P(Z=z|W \in \mathbf{w}, V=v) = \frac{P(Z=z,W \in \mathbf{w}, V=v)}{P(W \in \mathbf{w}, V=v)} \equiv \frac{p(z,\mathbf{w},v)}{p(\mathbf{w},v)}, \\
  &Q(t|a,\mathbf{w}, v) := 1/[1-S(t|a,\mathbf{w},v)]. 
\end{align}
To ensure that $Q$ is well-defined, throughout we assume that the cumulative incidence is positive at all $t \in \mathcal{T}$ within each stratum $(a,\mathbf{w}, v)$. 
For $(t',a',\mathbf{w}',z',v') \in \mathcal{T} \times \mathcal{A} \times \mathbf{W} \times \mathcal{Z} \times g(\mathcal{Z})$, we define
\begin{align} 
  \xi(u|t',a',\mathbf{w}',z') &:=
  -S(t'|a',\mathbf{w}',z')\int_0^{t'} \bigg[\frac{1(x \in ds, \delta=1, a=a', w \in \mathbf{w}', z=z')}{p_2(s,a',\mathbf{w}',z')} \label{eq:IF_suvfunc} \\
  &\hspace{4cm} - 1(x \geq s, a=a', w \in \mathbf{w}', z=z') \frac{p_1(ds, a',\mathbf{w}',z')}{[p_2(s, a',\mathbf{w}',z')]^2}\bigg], \nonumber\\
  h(u|\mathbf{w}',z',v')&:=  \frac{1(w \in \mathbf{w}', v=v')}{p(\mathbf{w}',v')}\left[1(z=z') - p(z'|\mathbf{w}',v')\right], \label{eq:IF_z} \\
  f_{\varphi}(u|t',a',\mathbf{w}',v')&:= 
  \frac{Q(t'|0,\mathbf{w}', v')}{\gamma(t|a',v')}
  \sum_{z' \in \mathcal{Z}}\Big\{
  -p(z'|\mathbf{w}',v')\xi(u|t',a',\mathbf{w}',z') \nonumber \\
  &\quad + \big[1-S(t'|a',\mathbf{w}', z')\big]h(u|\mathbf{w}',z',v')+\gamma(t'|a',v')\Big[ p(z'|\mathbf{w}',v')\xi(u|t',0,\mathbf{w}',z')\nonumber \\
  &\quad - \big[1-S(t'|0,\mathbf{w}',z')\big]h(u|\mathbf{w}',z',v') \Big]\Big\}. \label{eq:IF_logRR}
\end{align}

\begin{theorem} \label{Thm:joint_distribution_logRRs}
Given \ref{assump:indep_AZ_on_W}-\ref{assump:censoring}, the stratum $v$ and with $f_{\varphi}$ defined in \eqref{eq:IF_logRR}, 
\begin{align} \label{eq:asymp_linear_logRRs}
  &n^{1/2}\big(\varphi_n(t|a,v)-\varphi(t|a,v)\big)_{a \in [k]}
  = \Big(n^{-1/2}\sum_{i=1}^n f_{\varphi}(U_i|t,a,\mathbf{w}_{a},v)\Big)_{a \in [k]} + o_p(1).
\end{align}
Moreover, $\{n^{1/2}(\varphi_n(t|a,v)-\varphi(t|a,v))_{a \in [k]} : t\in\mathcal{T}\}$ weakly converges to a $k$-variate Gaussian process with mean zero and covariance function
\begin{align*}
  \bs{\Sigma}_{\varphi}(t,t^{\prime}|v)=E\big[(f_{\varphi}(U|t,a,\mathbf{w}_{a},v))_{a \in [k]} \otimes (f_{\varphi}(U|t',a,\mathbf{w}_{a},v))^\top_{a \in [k]}\big],  
\end{align*}
where $\otimes$ denotes the Kronecker product.
\end{theorem}
This theorem is an immediate consequence of Lemmas \ref{lemma:joint_distribution_CHFs}--\ref{lemma:joint_distribution_RRs} in the appendix, and so its proof is omitted.

\section{Efficiency gains from sharing controls}
\label{sec:efficiency_gain}

\subsection{Motivation}
In this section, we will demonstrate the statistical efficiency that can be gained from running a platform trial with a shared control arm as opposed to using separate control arms, as is done in more traditional clinical trial designs. Two forms of gains can be realized by using a platform design.

The first follows immediately from the design's use of a shared control arm. 
In particular, a platform trial conducted contemporaneously with and in the same population as separate, intervention-specific trials will generally attain the same statistical power for marginal evaluation of each active intervention as can the separate trials, while enrolling fewer participants on control. Indeed, due to the use of a shared control arm, fewer total participants can be enrolled in a platform trial than in separate, intervention-specific trials, while still maintaining the same sample sizes for the comparison of each active intervention versus control. This point is illustrated in Table~\ref{tab:randomization scheme_pairwise comparison}, which compares a 3-arm platform trial to two separate 2-arm trials. The 3-arm platform trial in that table corresponds to the same setting as was illustrated in Table~\ref{tab:randomization scheme}, except for being limited to only two interventions for simplicity and using ${\rm C}$ to label the shared control arm. The corresponding separate trials enroll the same number of individuals to each active intervention in the same windows, but enroll twice as many controls in windows where both active interventions are under randomization. 

\begin{table}[ht]
\centering
\begin{tabular}[t]{p{2cm}lllll}
  \toprule
  Trial & Window 1  &  Window 2  &  Window 3 & Window 4 & 
  Window 5\\
  \midrule \midrule
  \multirow{3}{*}{\parbox{2cm}{3-arm platform trial}} & 
    ${\rm A}_1:{\rm C}$ & ${\rm A}_1\hspace{0.35em}\rotatebox[origin=b]{270}{\large$\Rsh$}$ & ${\rm A}_1\hspace{0.35em}\rotatebox[origin=b]{270}{\large$\Rsh$}$ &  & ${\rm A}_1:{\rm C}$ \\
    &   & ${\rm A}_2:{\rm C}$ & ${\rm A}_2:{\rm C}$ & ${\rm A}_2:{\rm C}$ & \\[0.5em]
    & $(2n_1)$ & $(3n_2)$ & $(3n_3)$ & $(2n_4)$ & $(2n_5)$ \\
  \midrule    
  \multirow{3}{*}{\parbox{2cm}{2-arm separate trials}} &
   ${\rm A}_1:{\rm C}_1$ & ${\rm A}_1:{\rm C}_1$ & ${\rm A}_1:{\rm C}_1$ &  & ${\rm A}_1:{\rm C}_1$ \\
   &   & ${\rm A}_2:{\rm C}_2$ & ${\rm A}_2:{\rm C}_2$ & ${\rm A}_2:{\rm C}_2$ & \\[.5em]
   & $(2n_1)$ & $(4n_2)$ & $(4n_3)$ & $(2n_4)$ & $(2n_5)$ \\
  \bottomrule
\label{tab:randomization scheme_pairwise comparison}
\end{tabular}
\caption{Randomization schemes and the expected enrollment size per active intervention arm in each window under an illustrative 3-arm platform trial versus two 2-arm separate trials, where the expected enrollment size per arm in Window $j$ is denoted by $n_j$, $j=1,\ldots,5$, for both trials.}
\end{table}

The second advantage enjoyed by platform designs, which is a key finding of this work, involves a gain in efficiency for comparisons of the efficacy of different active interventions. Such comparisons are useful, for example, when aiming to evaluate the noninferiority of one intervention relative to another. Before providing theoretical insights as to the reasons for and generality of this gain, we present a simple numerical example illustrating how significant it can be in practice. To do this, we provide simulation results in a simple, binary outcome setting in which no covariates are measured. This setting can be seen to be a special case of the more general right-censored setup studied in this paper by letting $C=\tau$, $Y=1(T \le \tau)$ with $\tau$ denoting the end of follow-up, taking the covariate $Z$ to be a degenerate random variable that only takes the value 0, and taking the function $g$ used to define $V=g(Z)$ to be the identity function. Since $V$ is trivial and there is only one time point $\tau$ of interest, we write $RR(a)$, rather than $\gamma(\tau|a,v)$, to denote the relative risk of $a$ in this example. We compare the statistical power for various hypothesis tests based on data from 3-arm platform trials versus two separate 2-arm trials. 
The mean outcome on control is equal to 0.02, and the active interventions are such that the relative risk on intervention 1, namely $RR(1)$, is equal to 0.35 and the relative risk on intervention 2, namely $RR(2)$, is equal to 0.5. We evaluate statistical power in the following settings:
\begin{enumerate}[noitemsep,label=(\alph*), ref=(\alph*)]
  \item\label{it:standardPlat} a standard 3-arm platform trial with a single window in which both active interventions are under randomization and the sample size is selected to achieve 90\% power for marginal tests of $RR(a)\ge 1$ via 0.025-level Wald tests at the design alternative of $RR(a)=0.5$, $a\in\{1,2\}$,
  \item\label{it:separate} two separate 2-arm trials with sample sizes similarly selected for 90\% power, and
  \item\label{it:expandedPlat} an expanded 3-arm platform trial that contains a single randomization window as \ref*{it:standardPlat} does, but whose total sample size (controls and active interventions combined) is equal to the sum of the total sample sizes of the two separate trials in \ref*{it:separate}.
\end{enumerate}
The power calculations used to determine sample sizes for \ref*{it:standardPlat} and \ref*{it:separate} result in 1750 participants enrolled to each active intervention and 1750 participants enrolled to each control arm, namely the shared control in \ref*{it:standardPlat} and each separate control arm in \ref*{it:separate}. 
With marginal significance levels of 0.025, the null hypotheses to be tested are $RR(1) \ge 1$, $RR(2) \ge 1$, and $RR(1) \ge RR(2) + 0.1$, where Wald tests are used in all settings.
In Table \ref{tab:sim for uncensored case}, we present the power of rejecting these null hypotheses, along with the enrollment sizes of intervention and control arms in each dataset. 
Consistent with the earlier discussion regarding the equal power for marginal tests obtained by platform designs and separate trial designs that enroll the same number of participants to each active intervention, the platform trial \ref*{it:standardPlat} and the separate trial \ref*{it:separate} achieve the same power for the tests of $RR(1)\ge 1$ and $RR(2)\ge 1$. As increasing sample size increases power, it also follows naturally that the expanded platform trial \ref*{it:expandedPlat} attains higher power than \ref*{it:standardPlat}, and therefore \ref*{it:separate} as well, for these marginal tests. It is perhaps more surprising that both platform trials considered attain considerably higher power (19\%-28\% on an absolute scale) for testing $RR(1)\ge RR(2) + 0.1$ than do the separate trials in \ref*{it:separate}. This is true for \ref*{it:standardPlat} in spite of the fact that fewer total participants are enrolled in that trial than in the combined separate trials. In the remainder of this section, we provide analytical arguments establishing the generality of this improvement in power that platform designs enjoy for comparisons of the efficacies of different active interventions. When giving these arguments, we consider the general case where the outcome may be right-censored, covariates may be conditioned upon or adjusted for, and several interventions may be under randomization in any given window.

\begin{table}[ht]
\centering
\small\begin{tabular}[t]{lccccccc}
  \toprule
    &  \multicolumn{3}{c}{Sample size} &   & \multicolumn{3}{c}{Null} \\
  \cmidrule(r){2-4} \cmidrule(r){6-8}
    & total & controls & interventions &  & $RR(1) \ge 1$ & $RR(2) \ge 1$ 
    & $RR(1) \ge RR(2) + 0.1$ \\
  \midrule 
 \ref*{it:standardPlat} Platform & 5250 & 1750 & 3500 &  & 0.99 & 0.90 & 0.69 \\
 \ref*{it:separate} Separate & 7000 & 3500 & 3500 &  & 0.99 & 0.90 & 0.50 \\
 \ref*{it:expandedPlat} Expanded platform & 7000 & 2333 & 4667 &  & 1.00 & 0.95 & 0.78 \\
 \bottomrule
\label{tab:sim for uncensored case} 
\end{tabular}
\caption{Power of various tests under
\ref*{it:standardPlat} a standard 3-arm platform trial with a single window in which powered at 90\% for tests of the marginal nulls, \ref*{it:separate} two 2-arm separate trials with sample sizes similarly selected for 90\% power, and \ref*{it:expandedPlat} an expanded 3-arm platform trial with a single randomization window as in trial \ref*{it:standardPlat}, but whose total sample size (controls and interventions combined) is equal to the sum of the total sample sizes of the two separate trials in \ref*{it:separate}. The total sample size reflects the total number of participants enrolled in the platform trial or across the two separate trials. The arm-specific sample sizes listed reflect the expected number of participants to be on control (shared control for the platform trial, sum of the two control arms for the two separate trials) or active interventions (total participants on intervention 1 and intervention 2).}
\end{table}

\subsection{Theoretical guarantees}\label{sec:theoreticalGuarantees}
We consider the case where the goal is to compare the relative risks of two candidate interventions $a_1$ and $a_2$ within a stratum $v$ and at a given time point $t$. We allow this comparison to be made based on a differentiable contrast function $\Theta \colon (0,\infty)^2 \rightarrow \mathbb{R}$, and we refer to $\Theta(\gamma(t|a_1,v),\gamma(t|a_2,v))$ as the relative efficacy of interventions $a_1$ and $a_2$. For brevity, we let $\theta:= \Theta(\gamma(t|a_1,v),\gamma(t|a_2,v)))$, where the values of $t$, $a_1$, $a_2$, and $v$ are treated as fixed for the remainder of this subsection. 
Our analysis will apply to any contrast function $\Theta$ that satisfies the following condition.
\begin{assumpenum}[resume=assumptions]
  \item\label{assump:negative_product_firstderivatives} 
 $\left[\frac{\partial}{\partial r_1}\Theta(r_1,r_2)\right]\left[\frac{\partial}{\partial r_2}\Theta(r_1,r_2)\right]<0$ for all $(r_1,r_2)\in (0,\infty)^2$.
\end{assumpenum}
The above is satisfied by additive and multiplicative contrasts of the relative risks, namely $\Theta(r_1,r_2)=r_1-r_2$ and $\Theta(r_1,r_2)=r_1/r_2$. 
In Appendix~\ref{app:partialOrder}, we argue that the above condition is in fact natural whenever $\Theta$ is to be used to determine the superiority or noninferiority of $a_1$ relative to $a_2$.

To quantify the efficiency gains that can be realized by running a platform trial, we compare the widths of confidence intervals for $\theta$ based on data from two settings. In the first, the data $\{U_i=(X_i, \Delta_i, A_i, W_i, Z_i)\}_{i=1}^n$ arise as $n$ iid observations in a platform trial (see Section \ref{sec:data_structure and assumptions}). In the second, the pooled data from $k$ separate independent trials are used. Specifically, these pooled data take the form $\cup_{a \in [k]}\{U_{ai}=(X_{ai}, \Delta_{ai}, A_{ai}, Z_{i})\}_{i=1}^{m_{a}}$,
where $\{U_{ai}\}_{i=1}^{m_{a}}$ contains the data from the individual separate trial evaluating active intervention $a$. We suppose that each $\{U_{ai}\}_{i=1}^{m_{a}}$ is an iid sample from some distribution $P_a^\dagger$.   
The data structure observed in each separate trial is similar to that observed in the platform trial, except that no window variables are observed and $A_{ai}$ has support in $\{0_a,a\}$, where $0_a$ denotes the control arm in the separate trial for intervention $a$. 
As in the platform trial, the observed time $X_{ai}$ is the minimum of an event time $T_{ai}$ and a censoring time $C_{ai}$, and $\Delta_{ai}=1(T_{ai}\le C_{ai})$. We similarly suppose conditionally independent censoring and randomization, in this case that $C_{ai}$ is independent of $T_{ai}$ given $(A_{ai},Z_{ai})$ and $A_{ai}$ is independent of $Z_{ai}$. 
The overall size of this pooled dataset is $m = \sum_{a=1}^k m_{a}$.

In our theoretical analysis, we focus on the case where the platform trial and the separate trials are identical in all regards except for the fact that a shared control arm is used in the platform trial, whereas a different control arm is used in each of the separate trials. We, therefore, wish to ensure that the population enrolled, the efficacy of each intervention, and the distribution of the censoring and event times is similar across the two settings. To formalize this, we impose a condition relating the distribution $P_a^\dagger$ that gave rise to data in the separate trial for intervention $a$ to the conditional distribution of $U$ under $P$ conditionally on $A\in\{0,a\}$ and $W\in \mathbf{w}_{a}$. Below we denote this conditional distribution by $P(\,\cdot\,|A\in\{0,a\},W\in \mathbf{w}_{a})$ and use $\,\overset{d}{=}\,$ to mean equality in distribution.
\begin{assumpenum}[resume=assumptions]
    \item \label{assump:identical_distribution} \textit{Platform and separate trials enroll from the same population:} For all $a \in [k]$, $(X_a,\Delta_a,1(A_a=a),Z_a)\overset{d}{=}(X,\Delta,1(A=a),Z)$, where $(X_a,\Delta_a,A_a,Z_a) \sim P_a^\dagger$, and $(X,\Delta,A,W,Z) \sim P(\,\cdot\,|A\in\{0,a\},W\in \mathbf{w}_{a})$.
\end{assumpenum}
Under this condition, the relative risk of intervention $a$ as compared to control through time $t$, conditionally on covariate level $v$, is the same in the platform trial and in the separate trials. This condition also implies that the active intervention is assigned with probability $1/2$ in each separate trial.

We also require that the platform trial and separate trials provide similar relative precision for estimating $\gamma(t|a,v)$ and $\gamma(t|a',v)$ for any active interventions $a$ and $a'$. We formalize this notion in terms of the standard errors for these two quantities. In the separate trials, the standard error $\gamma(t|a,v)$ will be on the order of $m_a^{-1/2}$. In the platform trial, it will be on the order of $n_a^{-1/2}$, where $n_a:= nP(A\in\{0,a\},W\in\mathbf{w}_a)$ denotes the number of observations that are expected to be used in the evaluation of the relative risk of active intervention $a$. The following condition imposes that, asymptotically, the ratio of the (squared) standard errors between intervention $a$ and $a'$ be the same across the two trials.
\begin{assumpenum}[resume=assumptions]
  \item \label{assump:identical_ratios} \textit{Platform and separate trials have the same relative precision across interventions:} For all $a,a'\in [k]$, it holds that $m_a/m_{a'}=n_a/n_{a'}$.
\end{assumpenum}

We now exhibit an estimator for the relative risk $\gamma(t|a,v)$ based on the data from the separate trial for a given active intervention $a$. 
For a covariate level $z$, let $P_{ma}^\dagger(z|v) \equiv \sum_{i=1}^{m_{a}}1(Z_{ai}=z, V_{ai}=v)/\sum_{i=1}^{m_{a}}1(V_{ai}=v)$.
We estimate the relative risk through time $t$ via
\begin{align*}
  \gamma_m^\dagger(t|a,v):=\frac{1-\sum_{z \in \mathcal{Z}}S_{m}^\dagger(t|a,z)P_{ma}^\dagger(z|v)}{1-\sum_{z \in \mathcal{Z}}S_{m}^\dagger(t|0_a,z)P_{ma}^\dagger(z|v)},
\end{align*}
where, for $a'\in\{0_a,a\}$, $S_{m}^\dagger(t|a',z)$ is the stratified Kaplan-Meier estimator of the probability that $T>t$ within the stratum where $(A,Z)$ equals $(a',z)$ based on data from the separate trial for intervention $a$.

We now provide the forms of the confidence intervals for $\theta$ that we consider based on data from the platform trial separate trials. 
Each of these intervals is built based on asymptotic normality results that appear in the appendix. Lemma~\ref{lemma:limiting_dist_platform_trials} shows that the estimator $\theta_n:= \Theta(\gamma_n(t|a_1,v), \gamma_n(t|a_2,v))$ based on platform data satisfies $n^{1/2}\left[\theta_n-\theta\right] \rightsquigarrow \mathcal{N}(0,\sigma_{a_1,a_2}(t|v)^2)$,
where $\rightsquigarrow$ denotes convergence in distribution and the form of $\sigma_{a_1,a_2}(t|v)^2$ is given in \eqref{eq:covariance_VE_p}. 
Similarly, Lemma~\ref{lemma:limiting_dist_separate_trials} establishes that the estimator $\theta_m^\dagger:=\Theta(\gamma_{m}^\dagger(t|a_1,v), \gamma_{m}^\dagger(t|a_2,v))$ based on separate trial data satisfies $m^{1/2}\left[\theta_m^\dagger-\theta\right]\rightsquigarrow \mathcal{N}(0,\sigma^\dagger_{a_1,a_2}(t|v)^2)$,
where $\sigma^\dagger_{a_1,a_2}(t|v)^2$ is defined in \eqref{eq:covariance_VE_s}.  Let $\sigma_n$ and $\sigma_m^\dagger$ denote consistent estimators of $\sigma_{a_1,a_2}(t|v)$ and $\sigma^\dagger_{a_1,a_2}(t|v)$, respectively. Fix a significance level $\alpha\in(0,1)$ and let $z_{1-\alpha/2}$ denote the $1-\alpha/2$ quantile of a standard normal distribution. 
Using data from the platform trial, an asymptotically valid two-sided $100(1-\alpha)\%$ confidence interval for the relative efficacy $\theta$ is given by
\begin{align} \label{eq:Wald_CI_platform}
  \Big[ \theta_n \pm z_{1-\alpha/2}\,\frac{\sigma_n}{n^{1/2}} \Big].
\end{align}
Using data from the separate trials, an analogous interval is given by
\begin{align} \label{eq:Wald_CI_separate}
  \Big[ \theta_m^\dagger
  \pm z_{1-\alpha/2}\,\frac{\sigma_m^\dagger}{m^{1/2}} \Big].
\end{align}
The interval in \eqref{eq:Wald_CI_platform} is a Wald-type interval based on an estimator of $\theta$ that is asymptotically efficient within the nonparametric model where the only assumption made on the platform trial data-generating distribution is that intervention assignment is randomized (see \ref{assump:indep_AZ_on_W}), and the interval in \eqref{eq:Wald_CI_separate} is similarly based on an efficient estimator in the model where the only assumption made on the separate trial data-generating distributions is that intervention assignment is randomized. As a consequence, if the platform trial confidence interval in \eqref{eq:Wald_CI_platform} were to be asymptotically shorter than the separate trial confidence interval in \eqref{eq:Wald_CI_separate}, then the platform trial would enable a more efficient estimation of the contrast $\theta$. The following theorem shows that this is indeed the case. Below ${\rm plim}$ denotes a probability limit.


  
\begin{theorem} \label{Thm:improve_efficiency_sharing_control}
  Suppose that \ref{assump:indep_AZ_on_W}-\ref{assump:identical_ratios} hold and $m,n\rightarrow\infty$ in such a way that $n/m \rightarrow \kappa$ for some $\kappa>0$. Let $\rho:= 1/\sum_{a'=1}^k P(A\in\{0,a'\},W\in \mathbf{w}_{a'})$. Denote the widths of the intervals in \eqref{eq:Wald_CI_platform} and \eqref{eq:Wald_CI_separate} by $\omega_n$ and $\omega_m^\dagger$, respectively.
  \begin{enumerate}[label=\roman*)]
      \item If $\kappa > \rho$, then the platform trial interval is shorter asymptotically, that is, ${\rm plim}_{m,n}\,\omega_n/\omega_m^\dagger<1$.
      \item  If $\kappa = \rho$, then the platform trial interval is no longer asymptotically, that is, ${\rm plim}_{m,n}\,\omega_n/\omega_m^\dagger\le 1$.
  \end{enumerate}
\end{theorem}
The above theorem shows that, under conditions, running a platform trial will never harm the precision of a confidence interval contrasting two different interventions and will improve it in some cases.
To see the lack of harm, note that $\rho$ is never greater than one, and therefore Theorem \ref{Thm:improve_efficiency_sharing_control} implies that the confidence interval based on the platform trial data is never wider than that based on the separate trial data if $n=m$. In fact, a stronger conclusion holds: the platform trial confidence interval is asymptotically no wider than that separate trial confidence interval provided the expected number of participants that are enrolled to each active intervention $a$ is equal in these two settings --- that is, $nP(A=a)=m_{a}P_{a}^\dagger(A=a)$ for all $a\in [k]$; the conclusion can be shown to hold since $\kappa=\rho$ in these cases. Because the platform trial utilizes the shared control arm, enforcing the same number of expected participants enrolled on the active interventions allows for the platform trial to be smaller ($n<m$) whenever there is at least one window in the platform trial in which two or more interventions are under randomization. More generally, when there is at least one such window, the platform trial can yield \textit{shorter} confidence intervals even in cases where it enrolled fewer participants --- this is true, in particular, when $\kappa\in (\rho,1)$. In fact, as we show in our simulations, there are realistic cases where this gain in precision is considerable.

\section{Adaptive noninferiority test} \label{sec:intersection_test}
We now provide a testing procedure to investigate the noninferiority of a specified intervention -- assumed to be intervention 1 here without loss of generality -- as compared to the most efficacious of the other interventions. We call intervention 1 noninferior to intervention $a$ if $\gamma(t|1,v) < \gamma(t|a,v) + \varepsilon$, where $\varepsilon> 0$ is a specified noninferiority margin, and noninferior to the most efficacious of the other interventions if $\gamma(t|1,v) < \min_{a\not=1} \gamma(t|a,v) + \varepsilon$. If it is not \textit{a priori} known whether any of the other interventions are in fact efficacious, noninferiority alone is insufficient to determine that an intervention has clinically meaningful efficacy. To handle such cases, it is natural to further require that the relative risk of intervention 1 falls below some specified threshold $\delta$. This leads to a null hypothesis test of
\begin{align*}
    H_0 : \gamma(t|1,v) \ge \min\{ \delta,\, \min_{a \not= 1}\gamma(t|a,v) + \varepsilon \}\ \ \textnormal{ vs. }\ \ H_1 :\textnormal{not }H_0. 
\end{align*}
The null $H_0$ rewrites as a union of $k$ marginal null hypotheses. In particular, $H_0$ holds if and only if at least one of the following marginal nulls holds: $H_{00} : \gamma(t|1,v) \ge \delta$ or $H_{0a} : \gamma(t|1,v) \ge \gamma(t|a,v) + \varepsilon$ for some $a \in\{2,\ldots,k\}$.
Consequently, the alternative hypothesis $H_1$ corresponds to the intersection of the complementary marginal alternatives, namely $H_{1a} : \textnormal{not }H_{0a}$, with $a$ varying over $\{0,2,3,\ldots,k\}$.

Based on the above observations, the null $H_0$ can be tested at significance level $\alpha$ by running (unadjusted) $\alpha$-level tests of the marginal nulls $H_{0a}$ vs. $H_{1a}$, with $a$ varying over $\{0,2,3,\ldots,k\}$, and rejecting $H_0$ if and only if all of these $k$ marginal tests reject. 
This test of $H_0$ vs. $H_1$, which we refer to as an intersection test, necessarily controls the type I error asymptotically provided the marginal tests do so, since
\begin{align} \label{eq:type_I_error_control}
  &\limsup_n P^n(\mbox{reject } H_0) = \limsup_n P^n\big(\cap_{a \in \mathcal{A}_a}\{\textnormal{reject }H_{0a}\}\big)\le \limsup_n P^n(\mbox{reject }H_{0a})=\alpha,
\end{align}
where $P^n$ denotes the distribution of $n$ independent draws from $P$. One natural implementation of the intersection test, which is the one we employ in our simulations, bases the marginal tests off of Wald-type confidence intervals constructed based on Lemma \ref{lemma:joint_distribution_RRs} in the appendix. In particular, the marginal test of $H_{00}$ rejects if the upper bound of a two-sided $100(1-2\alpha)\%$ confidence interval for $\gamma(t|1,v)$ is smaller than $\delta$ and, for $a\ge 2$,
the marginal test of $H_{0a}$ rejects if the upper bound of a two-sided $100(1-2\alpha)\%$ confidence interval for $\gamma(t|1,v)-\gamma(t|a,v)$ is smaller than $\varepsilon$. 

\section{Numerical studies} \label{sec:numerical}
We now present numerical studies to imitate the evaluation of $k=10$ vaccines in a placebo-controlled platform trial versus in multiple separate trials.  
Within each window, enrollment is uniform over calendar time. 
Table \ref{tab:VE_AR_enrollment in windows for interventions}
summarizes the enrollment timelines of active interventions by windows (until 3 months post trial initiation), along with the window widths and the enrollment size per arm in each window. 
A four-category baseline variable $Z$ is measured for each participant, where the distribution of this variable depends on the enrollment window. In particular, $(P\{Z=z\mid W=w\})_{z=1}^4$ is equal to (0.1, 0.2, 0.3, 0.4) when $w\in \{1,2\}$ and (0.4, 0.3, 0.2, 0.1) otherwise. Within each stratum of $Z$, placebo participants have piecewise-constant hazard functions that change values only at the calendar times indicating a transition between windows (Months 1, 1.5, 2, and 2.5). The strata where $Z\in \{1,2\}$ are the lower-risk strata, and their hazard functions across windows (1, 2, 3, 4, 5) are such that the corresponding 6-month attack rates in the placebo arm are equal to (12\%, 12\%, 6\%, 4\%, 4\%). The strata where $Z\in \{3,4\}$ are the higher-risk strata, and the hazard functions are chosen so that the 6-month attack rates are twice those as in the lower-risk strata. Though the event is somewhat rare, all vaccine-vs.-shared-placebo comparisons include 150 events shortly after enrollment opens --- the median across Monte Carlo repetitions ranges from a median 2.5 months for intervention 4 to 4.5 months for intervention 5. The hazard ratio vaccine efficacy --- defined as one minus the hazard ratio of vaccine versus placebo recipients --- is presented in Table~\ref{tab:VE_AR_enrollment in windows for interventions}. For simplicity, this vaccine efficacy is made to be constant over time.  
The time for loss to follow-up is taken to follow a ${\rm Uniform}(0,120)$ distribution that is independent of all other variables under consideration, so that there is 10\% annual loss to follow-up during the study, which runs for a total of 18 calendar months.

We consider estimators in two cases: (i) reducing $V$ to a constant variable and estimating the marginal relative risk, and (ii) taking $V=Z$ and estimating the conditional relative risk. All results in the main text pertain to the estimation of the marginal relative risk, and results for estimators of the covariate-stratified relative risk are reported in Appendix~\ref{app-sec:supp_figures}. 
We compare the statistical efficiency of using data from a platform trial (total sample size of 40,400) rather than separate trials (total sample size of 69,800) in the estimation for relative risk ratios of intervention 7 versus other interventions at $t= 3$ or 6, under moderate loss to follow-up and administrative censoring at 6, 9, 12 or 18 calendar months. The efficiency gain is measured by confidence-interval-width ratios --- the lower the values of these ratios are, the more efficient the platform trial will be as compared to the separate trials. 

Similar coverage was observed for the confidence intervals based on platform trial data and separate trial data. 
In particular, the empirical percentiles (0\%, 25\%, 50\%, 75\%, 100\%) of the confidence interval coverage across the total of 72 scenarios considered were (91.8\%, 93.7\%, 94.5\%, 95\%, 95.7\%) for the platform trial and (91.4\%, 93.4\%, 94\%, 94.5\%, 95.6\%) for the separate trials.
Figure \ref{fig:hist_ciwidth_ratios} displays the efficiency gain of running a platform trial over running separate trials for the marginal and conditional estimands. Across all settings considered, analyses based on the platform trial were at least as efficient, and often more efficient, than those based on the separate trial data. This is true in spite of the fact that the platform trial enrolled over 40\% fewer participants than the separate trials. This result is consistent with our theoretical guarantees in Section~\ref{sec:theoreticalGuarantees}. Further details on these efficiency gains, broken down by evaluation time $t$ and administrative censoring time, can be found in Figures~\ref{fig:hist_cRRR_ciwidth_ratios} and \ref{fig:hist_mRRR_ciwidth_ratios} in the appendix.

\begin{table}[hbt!]
\centering
\begin{tabular}[t]{ccccccc}
\toprule
       & VE & Window 1  &  Window 2  &  Window 3 &  Window 4 & Window 5\\
Width  &    & (0, 1]    &  (1, 1.5]  &  (1.5, 2] &  (2, 2.5] & (2.5,  3]\\
Enrollment size per arm &   & 1000    &  900       &  1200      & 1400      & 1000 \\ 
\midrule
  Vaccine    &            &             &             &            &        &   \\
  1 & 0\%  &            & \checkmark  &             & \checkmark &  \checkmark  \\
  2 & 10\% & \checkmark &             &             &            &  \checkmark \\
  3 & 20\% & \checkmark &             &             & \checkmark &  \checkmark   \\
  4 & 30\% & \checkmark & \checkmark  & \checkmark  &            &          \\
  5 & 40\% &            & \checkmark  & \checkmark  &            &          \\
  6 & 50\% &            & \checkmark  & \checkmark  & \checkmark &          \\
  7 & 60\% & \checkmark & \checkmark  & \checkmark  &            & \checkmark  \\
  8 & 60\% & \checkmark &             &             & \checkmark & \checkmark  \\
  9 & 70\% & \checkmark & \checkmark  & \checkmark  & \checkmark &          \\
 10 & 70\% & \checkmark & \checkmark  & \checkmark  & \checkmark & \checkmark \\
\bottomrule
\end{tabular}
\caption{The hazard ratio
vaccine efficacy (VE) and windows of enrollment (with corresponding calendar times, in months) for the 10 vaccines in the platform trial.
}
\label{tab:VE_AR_enrollment in windows for interventions}
\end{table}



\begin{figure}[hbt!]
\begin{center}
 \includegraphics[trim={0cm 0.2cm 0.2cm 0cm}, clip, width=0.7\textwidth]
 {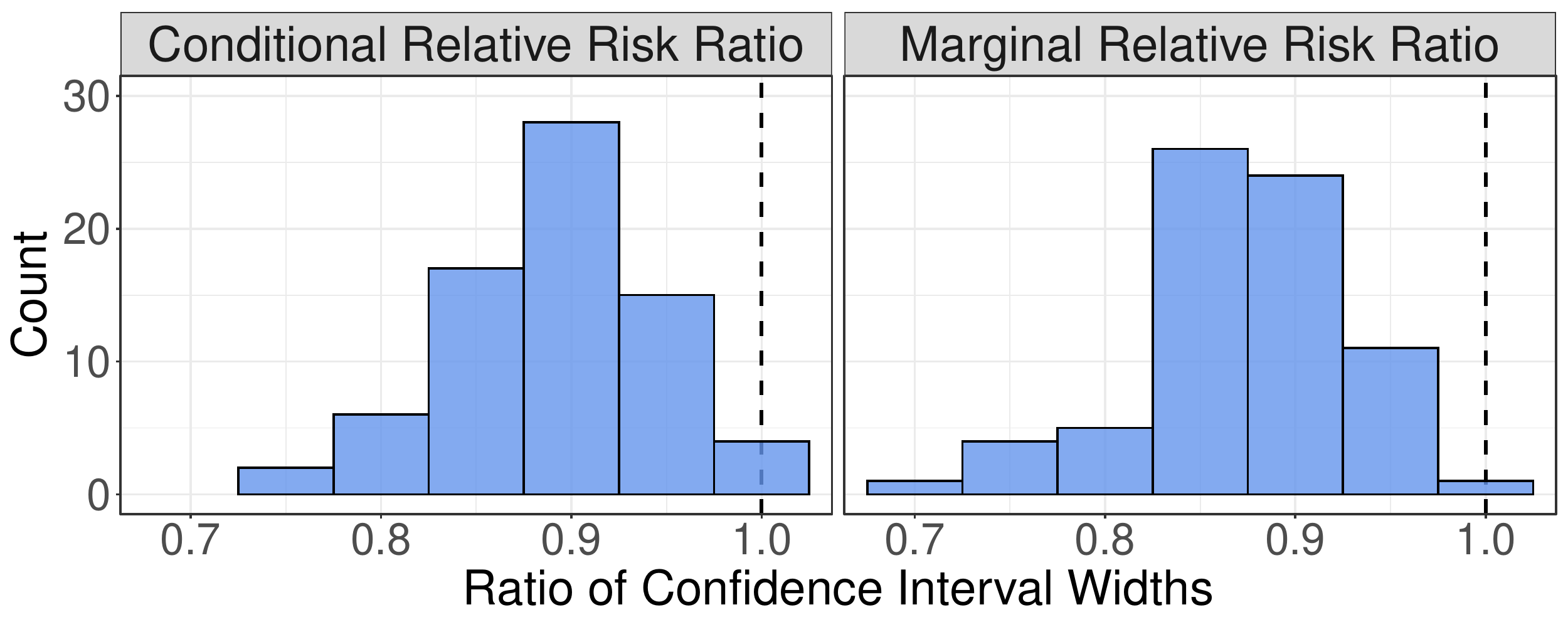}
 \caption{The ratios of the widths of confidence intervals based on data from the platform trial versus from separate trials for both marginal and conditional relative risks. Results for all comparator interventions (interventions other than the reference intervention, namely intervention 7), both evaluation times ($t= 3$ and 6), and all administrative censoring times (calendar months 6, 9, 12, and 18) are summarized in these histograms.}
 \label{fig:hist_ciwidth_ratios}
\end{center}
\end{figure}

We evaluated the performance of the adaptive noninferiority test described in Section~\ref{sec:intersection_test}. The evaluation of this test contains two scenarios: one with intervention 7 serving as the prespecified candidate whose noninferiority will be evaluated, and the other with intervention 9 serving as this candidate. Relative efficacy was quantified using the differences of the marginal relative risk of the pre-specified candidate and that of the other active interventions. We also evaluate the performance of a likelihood ratio type test, which is described in Appendix~\ref{app-sec:LRT}. 
 
We set the significance level at $\alpha = 0.025$ and the efficacy threshold at $\delta = 0.7$. 
In Figure \ref{fig:rejrates_tol0.7_ltfr10_comparison}, we report the empirical rejection rates (over 1000 Monte Carlo simulation runs) to evaluate the power and type I error control based on the data from the platform trial and the data from separate independent trials, with all the observations censored at $t=6$ and under moderate loss to follow-up. From the figure, we see that type I errors of both the intersection test and the likelihood ratio type test are conservatively controlled by the test data from either trial. The tests based on the data from the platform trial yields higher power than the one based on separate independent trials, and the intersection tests attain considerably higher power than the likelihood ratio type tests. 

We also evaluated an oracle noninferiority test that tested a simpler null hypothesis than the intersection test, namely the hardest-to-be-rejected null among the $k$ total marginal null hypothesis tests (of $H_{00}$ and $H_{0a'}$, $a' \not=$ the pre-specified intervention $a$) described in Section~\ref{sec:intersection_test}. The hardness of a marginal null hypothesis was quantified via the statistical power of testing this marginal null. The oracle test imitates an idealized setting where it is possible to set \textit{a priori} a single benchmark intervention for the noninferiority test. 
This setting is certainly unrealizable in practice since evaluating it relies on knowing the true operating characteristics of tests of the $k$ total marginal null hypotheses. As the intersection test must reject all $k$ of these null hypotheses, it will necessarily have lower power than the oracle. Nevertheless, the intersection test achieved only slightly lower power than did the oracle test in all considered scenarios. In particular, the power of the intersection test relative to that of the oracle test ranged from 0.4\% to 11.6\% lower on an absolute scale when intervention 7 was pre-selected and from 0\% to 4.3\% when intervention 9 was pre-selected.


\begin{figure}[hbt!]
\begin{center}
 \includegraphics[trim={0.1cm 0.4cm 0.2cm 0.1cm}, clip, width=0.65\textwidth]{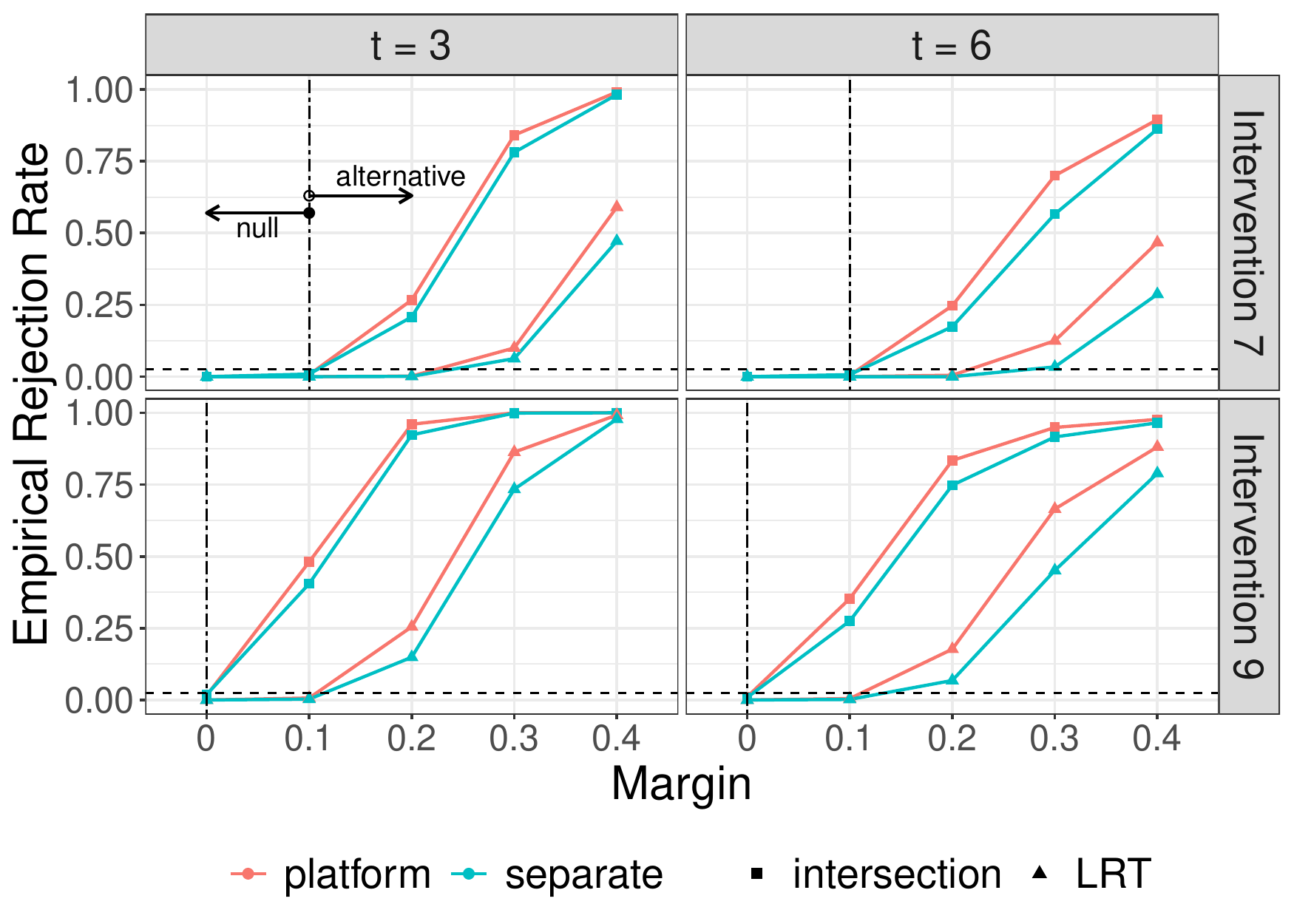}
 \caption{The empirical rejection rates of the intersection test and the likelihood ratio type test for different pre-selected active interventions based on data from the platform trial and separate trials, evaluated at $t=3, 6$ with varying margin values, the significance level of $\alpha=0.025$ and the efficacy threshold $\delta=0.7$, when all the observations are censored at $t=6$ and subject to moderate loss to follow-up.}
 \label{fig:rejrates_tol0.7_ltfr10_comparison}
\end{center}
\end{figure}


\section{Data illustration} \label{sec:data_analysis}
In two parallel Antibody Mediated Prevention (AMP) trials, referred to as HVTN 703 and HVTN 704, participants were randomly assigned in a 1:1:1 ratio to control group, low-dose intervention group, and high-dose intervention group \citep{corey2021two}. The primary endpoint was the post-trial days to type-1 human immunodeficiency virus (HIV-1) infection through the week 80 study visit. 
The annual loss to follow-up rates were 6.3\% in HVTN 703 and 9.4\% in HVTN 704. 
The two trials were conducted in different populations, with HVTN 703 enrolling at-risk women in sub-Saharan Africa and HVTN 704 enrolling at-risk cisgender men and transgender persons in the Americas and Europe. 
Following \cite{corey2021two}, we stratified our analyses by 
trial, thereby reporting separate prevention efficacy estimates for HVTN 703 and HVTN 704. 

To illustrate our developed testing procedure, we created open platform trial datasets by subsampling data from the parallel, multi-arm AMP trials as described below. 
A subset of the participants in the original trial were enrolled in the platform trials we considered. In these platform trials, we fixed the total number of participants enrolled to each active intervention to be approximately 320 in HVTN 703 and 450 in HVTN 704, which is approximately half the total number of participants enrolled to the active arms in the original trials. 
To enroll participants, the data for each trial were divided into four chronological windows, each containing a quarter of the sample; the resulting sample sizes for intervention groups over windows are presented in Figure \ref{fig:sample_dist_arms_windows_protocols} in the appendix. Participants were resampled, without replacement, from HVTN 703 or HVTN 704 to generate platform trial datasets that have a desired level of overlap in the control arm. More concretely, we sought to fix the proportion of controls shared at a value ranging from 0.25 to 0.50, where this proportion is defined as the number of controls that are shared by the two arms divided by the total number of controls enrolled in the platform trial. Details can be found in Supplementary Appendix \ref{app:dataIllustration}. 





We report inference for the relative efficacy of the low-dose intervention relative to that of the high-dose intervention (on an additive scale), where lower relative efficacy values indicate more favorable performance of the low-dose intervention. Because the trial results were already public when we received the data, we were not able to prespecify a noninferiority margin $\varepsilon$ or efficacy threshold $\delta$. Therefore, here we focus on reporting the widths of 95\% confidence intervals for the relative efficacy through week 80, averaged across the 
datasets considered. 
Figure~\ref{fig:noninferiority_ciwidths_stratified_low_doses} displays these widths.
In both trials, there is a trend towards more shared control data between two dose groups yielding tighter confidence intervals.


\begin{figure}[hbt!]
\begin{center}
 \includegraphics[trim={0.15cm 0.2cm 0.2cm 0.1cm}, clip, width=0.6\textwidth]{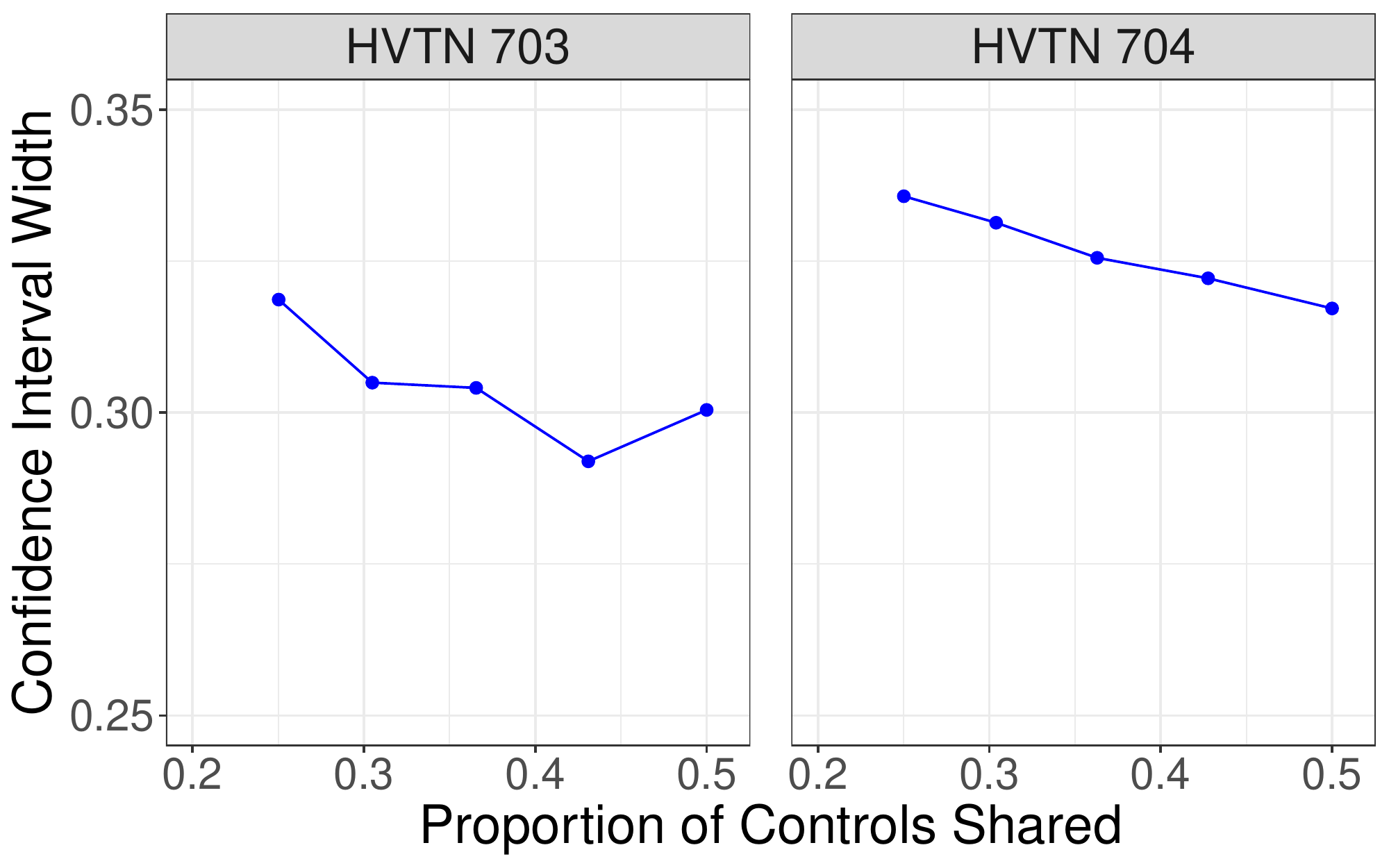}
 \caption{The confidence interval widths of the relative efficacy (on an additive scale) of the low-dose intervention relative to the high-dose intervention in HVTN 703 and HVTN 704, evaluated at week 80.} 
 \label{fig:noninferiority_ciwidths_stratified_low_doses}
\end{center}
\end{figure}


\section{Conclusion}\label{sec:conclusion}
This article established that platform trials can lead to more precise estimation of the relative efficacy of two different interventions. To this end, the joint distributions of estimators of the efficacy of multiple active interventions versus contemporaneous control were established. Such joint distributions were also shown to enable adaptive tests of noninferiority wherein the identity of the most efficacious comparator intervention need not be known in advance. Simulation studies were provided to support this asymptotic theory. In this study, a platform trial enrolled approximately 67\% as many participants as would have separate, independent trials and yet yielded up to 25\% narrower confidence intervals for contrasts of the efficacies of two different interventions.

This article only allows for the incorporation of discrete baseline covariates to account for the possibility of informative censoring. In future work, it would be of interest to extend our approach to incorporate continuous or high-dimensional discrete covariates. One possible area of future work in this direction would be to construct a targeted minimum loss-based estimator \citep{van2011targeted} that allows for the use of supervised machine learning methods to estimate the needed conditional probabilities that involve these covariates.

Though the positive correlation that we established for the joint distribution of the efficacy estimates across the different arms is advantageous for contrasting the efficacy of different interventions, it also has a disadvantage. To see why, note that, due to the use of a single shared placebo arm, there is a possibility that, by random chance, an unusually high or low number of events will be observed on the placebo arm. As a consequence, many of the arm-specific efficacy estimators may provide unusually high or low estimates of efficacy. This can, in turn, lead to scenarios where many interventions are falsely suggested to be efficacious 
or inefficacious \citep{howard2018recommendations}.  
Though such challenges may be avoided by employing an appropriate multiple testing correction, such procedures have previously been viewed as undesirable for platform trials that involve many stakeholders since they can disincentivize participation in these economical designs in favor of conducting more costly, intervention-specific trials. Therefore, as mentioned in the Introduction, recent platform trials have followed multi-armed trials in not using a multiplicity adjustment \citep{WHOSolidarityTrialConsortium2021,howard2021platform}. An alternative means to avoid elevated false (non-)discovery rates in platform designs is to simply increase the size of the placebo arm. Rather than reducing the correlation between arm-specific estimates, this instead directly reduces the variance of these estimates. Since the efficiency gains that we showed in this paper allowed the platform trials to be substantially smaller than the corresponding separate trials, the placebo arm could be substantially increased in size while still resulting in a smaller-sized trial than would have been two separate trials and improved precision for comparing the efficacy of multiple interventions. In any given setting, pre-trial Monte Carlo studies can be conducted to determine whether modifying the platform trial allocation ratio in this manner yields a design with the investigators' desired operating characteristics.

\section*{Acknowledgements}
The authors are grateful to Peter Gilbert for helpful discussions and to the participants and sponsors of the HVTN 703 and 704 trials. This work was supported by the NIH through award numbers DP2-LM013340 and 5UM1AI068635-09. The
content is solely the responsibility of the authors and does not necessarily represent the official views of
the NIH.

\newpage
\section*{Appendix}
\begin{appendix}

\section{List of Investigators in the AMP Investigator Group}

\textit{\textbf{AMP Leadership:}} 

\noindent Lawrence Corey, Myron Cohen, Srilatha Edupuganti, Nyaradzo Mgodi, Shelly Karuna, Philip Andrew
\\[8pt]
\noindent 
\textit{\textbf{HVTN 703/HPTN 081 Site Investigators of Record:}} 

\noindent \textbf{Botswana: Gaborone CRS} - Joseph Makhema. \textbf{Kenya: KISUMU CRS} – Grace Mboya. \textbf{Malawi: Blantyre CRS} –  Johnstone Kumwenda; \textbf{Malawi CRS (Lilongwe)} – Mina Hosseinipour. \textbf{South Africa: Aurum Institute Klerksdorp CRS} - Craig Innes; \textbf{Botha’s Hill CRS} – Elizabeth Spooner; \textbf{CAPRISA eThekwini CRS} – Nigel Garrett; \textbf{Chatsworth CRS} – Logashvari Naidoo; \textbf{Groote Schuur HIV CRS} – Catherine Orrell; \textbf{Rustenburg CRS} – William Brumskine; \textbf{Soweto HVTN CRS} – Fatima Laher; \textbf{Synexus Stanza (Mamelodi) Clinical Research Centre} – Sheena Kotze; \textbf{Vulindlela CRS} – Halima Dawood; \textbf{Wits RHI Ward 21 CRS} – Sinead Delany-Moretlwe.  \textbf{Zimbabwe: Milton Park CRS (formerly Parirenyatwa CRS)} – Pamela Mukwekwerere; \textbf{Seke South CRS} – Portia Hunidzarira; \textbf{Spilhaus CRS} – Shorai Mukaka. 
\\[8pt]
\noindent \textit{\textbf{HVTN 704/HPTN 085 Site Investigators of Record:}}

\noindent \textbf{United States: Alabama CRS} – Paul Goepfert; \textbf{Bridge HIV CRS} – Susan Buchbinder; \textbf{Bronx Prevention Center CRS} – Jessica Justman;  \textbf{Case Clinical Research Site} – Jeffrey M. Jacobson; \textbf{Chapel Hill CRS} – Cynthia Gay; \textbf{Columbia P\&S CRS} – Magdalena E. Sobieszczyk; \textbf{Fenway Health CRS} – Ken Mayer; \textbf{George Washington Univ. CRS} – Marc Siegel; \textbf{Harlem Prevention Center CRS} – Sharon Mannheimer; \textbf{New Jersey Medical School Clinical Research Center CRS} – Shobha Swaminathan; \textbf{New York Blood Center CRS} – Hong Van Tieu; \textbf{Penn Prevention CRS} – Ian Frank; \textbf{Seattle Vaccine and Prevention CRS} – Juliana McElrath; \textbf{The Hope Clinic of the Emory Vaccine Center CRS} – Srilatha Edupuganti; \textbf{The Ponce de Leon Center CRS} – Carlos del Rio; \textbf{UCLA Vine Street Clinic CRS} – Jesse Clark; \textbf{University of Rochester Vaccines to Prevent HIV Infection CRS} – Michael Keefer. \textbf{Peru: ACSA CRS} –  Juan Carlos Hinojosa; \textbf{Asociacion Civil Impacta Salud y Educacion, Barranco CRS} – Javier R. Lama; \textbf{Centro de Investigaciones Tecnológicas, Biomédicas y Medioambientales CRS - UNMSM} – Jorge Sanchez; \textbf{Asociacion Civil Impacta Salud y Educacion, San Miguel CRS} – Pedro Gonzales; \textbf{Via Libre CRS} – Robinson Cabello.

\section{Proofs}
\subsection{Preliminaries and notations}\label{app-sec:prelim_notations} 
The lemma below is useful for establishing \eqref{eq:estimand}.
\begin{lemma}\label{lem:covAdjEstimand}
 Under \ref{assump:indep_AZ_on_W}, the following hold for every time $t$, active intervention $a$, and possible realization $v$ of $V$:
 \begin{align*}
  P(T\le t|A=0, W \in \mathbf{w}_a, V=v)&= 1-\sum_{z\in\mathcal{Z}} S(t|0,\mathbf{w}_a,z) \,P(Z=z|W\in\mathbf{w}_a,V=v),\\
  P(T\le t|A=a, W \in \mathbf{w}_a, V=v)&= 1-\sum_{z\in\mathcal{Z}} S(t|a,\mathbf{w}_a,z) \,P(Z=z|W\in\mathbf{w}_a,V=v).
 \end{align*}
\end{lemma}
\begin{proof}[Proof of Lemma~\ref{lem:covAdjEstimand}]
We prove the first equality. The proof of the second equality is essentially identical. As shown in the below display, the first equality follows by the law of total expectation, \ref{assump:indep_AZ_on_W}, and the definition of $S$:
\begin{align*}
    &P(T\le t|A=0, W \in \mathbf{w}_a, V=v) = E\left[P(T\le t|A=0, W \in \mathbf{w}_a, Z)\middle|A=0, W \in \mathbf{w}_a, V=v\right] \\
    &\quad = E\left[P(T\le t|A=0, W \in \mathbf{w}_a, Z)\middle|W \in \mathbf{w}_a, V=v\right]
    = 1-E\left[S(t|0,\mathbf{w}_a,Z)\middle|W \in \mathbf{w}_a, V=v\right] \\
    &\quad = 1-\sum_{z\in\mathcal{Z}} S(t|0,\mathbf{w}_a,z) \,P(Z=z|W\in\mathbf{w}_a,V=v).
\end{align*}
\end{proof}

Here we collect all the needed notations that will appear in the remainder of this section. Let $p_1(t,a,\mathbf{w},z):=P(X \leq t, \Delta=1, A=a, W \in \mathbf{w}, Z=z)$, $p_2(t,a,\mathbf{w},z):=P(X \geq t, A=a, W \in \mathbf{w}, Z=z)$, and 
\begin{align} \label{eq:IF_Lambda}
  &f_{\Lambda}(u|t',a',\mathbf{w}',z') :=
  \int_0^{t'} \bigg[\frac{1(x \in ds, \delta=1, a=a', w \in \mathbf{w}', z=z')}{p_2(s,a',\mathbf{w}',z')}\\
  &\hspace{5cm} - 1(x \geq s, a=a', w \in \mathbf{w}', z=z') \frac{p_1(ds,a',\mathbf{w}',z')}{[p_2(s,a',\mathbf{w}',z')]^2}\bigg]. \nonumber
\end{align}
For $a,a' \in [k]$ for which $a \not= a'$, we also define
\begin{align} \label{eq:notations}
  \sigma^2_{\mathbf{w}_{a}, 0}(t,t^{\prime}|z) &:= E[f_{\Lambda}(U|t,0,\mathbf{w}_{a},z)f_{\Lambda}(U|t^{\prime},0,\mathbf{w}_{a},z)], \nonumber \\
  \sigma^2_{\mathbf{w}_{a},a}(t,t^{\prime}|z) &:=E[f_{\Lambda}(U|t,a,\mathbf{w}_{a},z)f_{\Lambda}(U|t^{\prime},a,\mathbf{w}_{a},z)], \nonumber \\
  \vartheta_{\mathbf{w}_{a}\mathbf{w}_{a'},0}(t,t^{\prime}|z) &:= E[f_{\Lambda}(U|t,0,\mathbf{w}_{a},z)f_{\Lambda}(U|t^{\prime},0,\mathbf{w}_{a'},z)].
\end{align}
Note that $\sigma^2_{\mathbf{w}_{a}, 0}(t,t^{\prime}|z) = \sigma^2_{\mathbf{w}_{a}, 0}(t^{\prime},t|z)$, but $\vartheta_{\mathbf{w}_{a}\mathbf{w}_{a'}, 0}(t,t^{\prime}|z) \not= \vartheta_{\mathbf{w}_{a}\mathbf{w}_{a'},0}(t^{\prime},t|z)$ for $t \not= t^{\prime}$.

Throughout we use $\otimes$ to denote the Kronecker product and $\mathcal{G}_k(\bs{0}_k,\sigma^2)$ to denote a $k$-variate mean-zero Gaussian process in some space $\ell^\infty(\mathcal{B})$ with covariance matrix function $\sigma^2 : \mathcal{B} \times \mathcal{B} \rightarrow \mathbb{R}^{k \times k}$, where the value of $\mathcal{B}$ will be clear from context. We also let $\mathcal{G}=\mathcal{G}_1$. We also recall that, for a given covariate value $z$, we let $v=g(z)$ denote a coarsening of $z$.

\subsection{Proofs for Section \ref{sec:inf_asymp_RR}}
\begin{lemma} \label{lemma:joint_distribution_CHFs}
Fix $z\in\mathcal{Z}$ and suppose that \ref{assump:censoring} holds. For all $t \in \mathcal{T}$,
\begin{align*}
  &n^{1/2}\big(\Lambda_n(t|0,\mathbf{w}_{a},z)-\Lambda(t|0,\mathbf{w}_{a},z)\big)_{a \in [k]}
  = \Big(n^{-1/2}\sum_{i=1}^nf_{\Lambda}(U_i|t,0,\mathbf{w}_{a},z)\,
    \Big)_{a \in [k]} + o_p(1), \nonumber \\
  &n^{1/2}\big(\Lambda_n(t|a,\mathbf{w}_{a},z)-\Lambda(t|a,\mathbf{w}_{a},z)\big)_{a \in [k]} =
   \Big(n^{-1/2}\sum_{i=1}^nf_{\Lambda}(U_i|t,a,\mathbf{w}_{a},z)
   \Big)_{a \in [k]} + o_p(1).
\end{align*}
In addition, the processes $\{n^{1/2}(\Lambda_n(t|0,\mathbf{w}_{a},z)-\Lambda(t|0,\mathbf{w}_{a},z))_{a \in [k]} : t \in \mathcal{T}\}$ and $\{n^{1/2}(\Lambda_n(t|a,\mathbf{w}_{a},z)-\Lambda(t|a,\mathbf{w}_{a},z))_{a \in [k]} : t \in \mathcal{T}\}$ weakly converge to a $k$-variate Gaussian processes with mean zero and covariance functions
\begin{align*}
  \bs{\Sigma}_{\Lambda,0}(t,t^{\prime}|z) & = E\big[(f_{\Lambda}(U|t,0,\mathbf{w}_{a},z))_{a \in [k]} \otimes (f_{\Lambda}(U|t',0,\mathbf{w}_{a},z))^\top_{a \in [k]}\big]\ \textnormal{ and }\  \\
  \bs{\Sigma}_{\Lambda}(t,t^{\prime}|z) & = E\big[(f_{\Lambda}(U|t,a,\mathbf{w}_{a},z))_{a \in [k]} \otimes (f_{\Lambda}(U|t',a,\mathbf{w}_{a},z))^\top_{a \in [k]}\big],
\end{align*}
respectively. 
\end{lemma}
Later we will use that, due to the independence of the participants on all active intervention arms and the fact that independence implies uncorrelatedness, $\bs{\Sigma}_{\Lambda}(t,t^{\prime}|z)$ reduces to ${\rm Diag}\big((\sigma^2_{\mathbf{w}_a, a}(t,t^{\prime}|z))_{a \in [k]}\big)$ for time $t$ and $t^{\prime}$, with
$\sigma^2_{\mathbf{w}_{a},a}(t,t^{\prime}|z)$ defined in \eqref{eq:notations}.
\begin{proof}[Proof of Lemma~\ref{lemma:joint_distribution_CHFs}]
Below we present the proof for the asymptotic results of the estimated control arm stratified CHFs, and the analogous arguments apply for the proof of that of the estimated intervention-specific stratified CHFs.
Recall that $N_n(t;0,\mathbf{w},z) = n^{-1}\sum_{i=1}^n1(X_i \leq t, \Delta_i=1, A_i=0, W_i \in \mathbf{w}, Z_i=z),$ and $Y_n(t;0,\mathbf{w},z) = n^{-1}\sum_{i=1}^n1(X_i \geq t, A_i=0, W_i \in \mathbf{w}, Z_i=z)$.
Noting that
\begin{align*}
  &p_1(t,0,\mathbf{w},z) = P(X \leq t, \Delta=1, A=0, W \in \mathbf{w}, Z=z); \nonumber \\
  &p_2(t,0,\mathbf{w},z) = P(X \geq t, A=0, W \in \mathbf{w}, Z=z),
\end{align*}
we readily see that $E[N_n(t; 0,\mathbf{w},z)] = p_1(t,0,\mathbf{w},z)$ and
$E[Y_n(t; 0,\mathbf{w},z)] = p_2(t,0,\mathbf{w},z)$.
Let $U_i = (X_i, \Delta_i, A_i, W_i, Z_i, V_i)$ denote the $i$-th replication of $U = (X, \Delta, A, W, Z, V)$. We have that
\begin{align}\label{eq:asymp_linear_counting_process}
  \big(N_n(t; 0, \mathbf{w}_{a}, z)-E[N_n(t; 0, \mathbf{w}_{a}, z)]\big)_{a \in [k]}
  &=\Big(n^{-1}\sum_{i=1}^nf_{N}(U_i|t,0,\mathbf{w}_{a},z)\Big)_{a \in [k]}, \nonumber \\
  \big(Y_n(t; 0, \mathbf{w}_{a}, z)-E[Y_n(t; 0, \mathbf{w}_{a}, z)]\big)_{a \in [k]}
  &=\Big(n^{-1}\sum_{i=1}^nf_{Y}(U_i|t,0,\mathbf{w}_{a},z)\Big)_{a \in [k]},
\end{align}
where for $u = (x,\delta,a,w,z,v)$ and $(t',\mathbf{w}',z') \in \mathcal{T} \times \mathbf{W} \times \mathcal{Z}$,
\begin{align*}
  &f_{N}(\cdot|t',0,\mathbf{w}',z') \colon u \mapsto 1(x \leq t', \delta=1, a=0, w \in \mathbf{w}', z=z')-p_1(t',0,\mathbf{w}',z'),\\
  &f_{Y}(\cdot|t',0,\mathbf{w}',z') \colon u \mapsto 1(x \geq t', a=0, w \in \mathbf{w}', z=z')-p_2(t',0,\mathbf{w}',z').
\end{align*}

From the definition of $\Lambda_n(\cdot|0,\mathbf{w},z)$ for any given window set $\mathbf{w}$ and the asymptotic linearity in \eqref{eq:asymp_linear_counting_process}, 
together with the Hadamard differentiability of $\Lambda$, applying $\alpha_n\beta_n-\alpha\beta = (\alpha_n-\alpha)\beta + (\beta_n-\beta)\alpha_n$ and
the functional delta method implies the following asymptotic linear approximation:
\begin{align*}
  \big(\Lambda_n(t|0,\mathbf{w}_{a}, z)-\Lambda(t|0,\mathbf{w}_{a}, z)\big)_{a \in [k]}
  =\Big(n^{-1}\sum_{i=1}^nf_{\Lambda}(U_i|t,0,\mathbf{w}_{a},z)\Big)_{a \in [k]} + o_p(1/n^{1/2}),       
\end{align*}
where $f_{\Lambda}$ has been defined in \eqref{eq:IF_Lambda} and the $o_p(n^{-1/2})$ term above converges to zero in probability uniformly over $t\in\mathcal{T}$.
Then we can define a class $\mathcal{F}=\{f_{\Lambda}(\cdot|t',a',\mathbf{w}',z') : (t',a',\mathbf{w}',z') \in \mathcal{T} \times \mathcal{A} \times \mathbf{W} \times \mathcal{Z}\}$. By Example 19.11 in \cite{Vaart2000}, $\mathcal{F}$ is a $P$-Donsker class. Therefore Donsker's theorem implies that
\begin{align*}
  \left\{n^{1/2}\big(\Lambda_n(t|0,\mathbf{w}_{a},z)-\Lambda(t|0,\mathbf{w}_{a},z)\big)_{a \in [k]} : t\in\mathcal{T}\right\}
  \rightsquigarrow \mathcal{G}_k(\bs{0}_k, \bs{\Sigma}_{\Lambda, 0}(t,t^{\prime}|z)).
\end{align*}
\end{proof}

Based on the developed asymptotic distribution of the estimated control arm stratified CHFs and that of the estimated intervention-specific stratified CHFs in Lemma \ref{lemma:joint_distribution_CHFs},
below we further establish the limiting distribution of the estimated stratified cumulative hazard ratios of the active interventions to their corresponding shared control. In this result, we let $(t,t^{\prime})\mapsto \bs{\Sigma}_{\Lambda R}(t,t^{\prime}|z)$ denote a covariance function that returns a $k \times k$ covariance matrix with $a$-th diagonal element 
\[
  \frac{1}{\Lambda(t|0,\mathbf{w}_{a},z)\Lambda(t^{\prime}|0,\mathbf{w}_{a},z)}\sigma^2_{\mathbf{w}_{a},a}(t,t^{\prime}|z) + 
  \frac{\Lambda(t|a,\mathbf{w}_{a},z)\Lambda(t^{\prime}|a,\mathbf{w}_{a},z)}{\Lambda^2(t|0,\mathbf{w}_{a},z)\Lambda^2(t^{\prime}|0,\mathbf{w}_{a},z)}\sigma^2_{\mathbf{w}_{a},0}(t,t^{\prime}|z),
\]
and off-diagonal element at the $(a,a^{\prime})$-th entry
\[ 
 \frac{\Lambda(t|a,\mathbf{w}_{a},z)\Lambda(t^{\prime}|a',\mathbf{w}_{a'},z)}{\Lambda^2(t|0,\mathbf{w}_{a},z)\Lambda^2(t^{\prime}|0,\mathbf{w}_{a'},z)}\vartheta_{\mathbf{w}_{a}\mathbf{w}_{a'}, 0}(t, t^{\prime}|z).
\]
\begin{lemma} \label{lemma:joint_distribution_CHRs}
  Fix $z\in\mathcal{Z}$. Under \ref{assump:censoring},
  \begin{align*}
    &n^{1/2}\big(\Lambda_n(t|a,\mathbf{w}_{a},z)/\Lambda_n(t|0,\mathbf{w}_{a},z)-\Lambda(t|a,\mathbf{w}_{a},z)/\Lambda(t|0,\mathbf{w}_{a},z)
    \big)_{a \in [k]}
  \end{align*} weakly converges to
  a $k$-variate Gaussian process with mean zero and covariance function $\bs{\Sigma}_{\Lambda R}(\cdot,\cdot|z)$.
\end{lemma}
\begin{proof}[Proof of Lemma~\ref{lemma:joint_distribution_CHRs}]
Based on the asymptotic linearity of $\Lambda_n-\Lambda$ in Lemma \ref{lemma:joint_distribution_CHFs}, the functional delta method gives that
\begin{align}\label{eq:asymp_linear_hazard_ratios}
  &\bigg(\frac{\Lambda_n(t|a,\mathbf{w}_{a},z)}{\Lambda_n(t|0,\mathbf{w}_{a},z)}-\frac{\Lambda(t|a,\mathbf{w}_{a},z)}{\Lambda(t|0,\mathbf{w}_{a},z)}\bigg)_{a \in [k]} \\
  &=\bigg(n^{-1}\sum_{i=1}^n\bigg[\frac{1}{\Lambda(t|0,\mathbf{w}_{a},z)}f_{\Lambda}(U_i|t,a,\mathbf{w}_{a},z)
  -\frac{\Lambda(t|a,\mathbf{w}_{a},z)}{\Lambda^2(t|0,\mathbf{w}_{a},z)}f_{\Lambda}(U_i|t,0,\mathbf{w}_{a},z)\bigg] \bigg)_{a \in [k]}
  + o_p(n^{-1/2}),   \nonumber
\end{align}
where the $o_p(n^{-1/2})$ term above converges to zero in probability uniformly over $t \in \mathcal{T}$. 
Let $\sigma^2_{\mathbf{w}_{a},0}(t,t^{\prime}|z)$ and $\vartheta_{\mathbf{w}_{a}\mathbf{w}_{a'}, 0}(t, t^{\prime}|z)$ be as defined in \eqref{eq:notations}, for $a,a'\in [k]$ and $a \not= a'$.
Therefore Slutsky's lemma, along with Lemma \ref{lemma:joint_distribution_CHFs}, implies that
\begin{align*}
  \bigg\{\bigg(\frac{\Lambda_n(t|a,\mathbf{w}_{a},z)}{\Lambda_n(t|0,\mathbf{w}_{a},z)}-\frac{\Lambda(t|a,\mathbf{w}_{a},z)}{\Lambda(t|0,\mathbf{w}_{a},z)}\bigg)_{a \in [k]}: t \in \mathcal{T}\bigg\} \rightsquigarrow \mathcal{G}_k(\bs{0}_k, \bs{\Sigma}_{\Lambda, R}(t,t^{\prime}|z)).
\end{align*}
\end{proof}

The next lemma shows that $S_n(t|0,\mathbf{w}_{a},z)$ is asymptotically linear with influence function $\xi$, where $\xi$ is as defined in \eqref{eq:IF_suvfunc}.
\begin{lemma} \label{lemma:asymp_surv_functions}
  Fix $z\in\mathcal{Z}$ and suppose that \ref{assump:censoring} holds. For any $t\in\mathcal{T}$,
  \begin{align*}
   &n^{1/2}\big(S_n(t|0,\mathbf{w}_{a},z)-S(t|0,\mathbf{w}_{a},z)\big)_{a \in [k]} = n^{-1/2}\Big(\sum_{i=1}^n\xi(U_i|t,0,\mathbf{w}_{a},z) \Big)_{a \in [k]} + o_p(1), \\
   &n^{1/2}\big(S_n(t|a,\mathbf{w}_{a},z)-S(t|a,\mathbf{w}_{a},z)\big)_{a\in [k]}
   =n^{-1/2}\Big(\sum_{i=1}^n\xi(U_i|t,a,\mathbf{w}_{a},z) \Big)_{a \in [k]} + o_p(1).
\end{align*}
Moreover, $\{n^{1/2}\big(S_n(t|0,\mathbf{w}_{a},z)-S(t|0,\mathbf{w}_{a},z)\big)_{a\in [k]} : t \in \mathcal{T}\}$ and $\{n^{1/2}\big(S_n(t|a,\mathbf{w}_{a},z)-S(t|a,\mathbf{w}_{a},z)\big)_{a\in [k]} : t \in \mathcal{T}\}$ weakly converge to $k$-variate Gaussian processes with mean zero and covariance functions
\begin{align*}
  \bs{\Sigma}_{S_0}(t,t^{\prime}|z) = E\big[(\xi(U|t,0,\mathbf{w}_{a},z))_{a \in [k]} \otimes (\xi(U|t',0,\mathbf{w}_{a},z))^\top_{a \in [k]}\big]\ \textnormal{ and }\ \\
  \bs{\Sigma}_{S}(t,t^{\prime}|z) = E\big[(\xi(U|t,a,\mathbf{w}_{a},z))_{a \in [k]} \otimes (\xi(U|t',a,\mathbf{w}_{a},z))^\top_{a \in [k]}\big],
\end{align*}
respectively.
\end{lemma}
\begin{proof}[Proof of Lemma~\ref{lemma:asymp_surv_functions}]
Applying the functional delta method to the results in Lemma \ref{lemma:joint_distribution_CHFs} yields the first display in the statement,
where, with $\sigma^2_{\mathbf{w}_{a},0}(t,t^{\prime}|z)$ and $\vartheta_{\mathbf{w}_{a}\mathbf{w}_{a'}, 0}(t, t^{\prime}|z)$ as defined in \eqref{eq:notations}, for $a,a' \in [k]$ and $a \not= a'$,
$\bs{\Sigma}_{S_0}(t,t^{\prime}|z)$ has
the $a$-th diagonal element $S(t|0,\mathbf{w}_{a},z)S(t^{\prime}|0,\mathbf{w}_{a},z)\sigma^2_{\mathbf{w}_{a},0}(t,t^{\prime}|z)$, and the off-diagonal element at the $(a, a^{\prime})$-th entry
$S(t|0,\mathbf{w}_{a},z)S(t^{\prime}|0,\mathbf{w}_{a'},z)\vartheta_{\mathbf{w}_{a}\mathbf{w}_{a'},0}(t,t^{\prime}|z)$. Similarly, the second display holds,
where $\bs{\Sigma}_{S}(t,t^{\prime}|z)$ has zero off-diagonal elements and
for $a\in [k]$, the $a$-th diagonal element $S(t|a,\mathbf{w}_{a},z)S(t^{\prime}|a,\mathbf{w}_{a}, z)\sigma^2_{\mathbf{w}_{a}, a}(t,t^{\prime}|z)$ with $\sigma^2_{\mathbf{w}_{a},a}(t,t^{\prime}|z)$ as defined in \eqref{eq:notations}.
\end{proof}

Below we introduce relevant notations for the development of the asymptotic linearity of $\gamma_n(t|a,v)$. Let $u=(x, \delta, a, w, z, v)$. 
Using notations $p(z'|\mathbf{w}',v')$, $\xi(\cdot|t',a',\mathbf{w}',z')$ and $h(\cdot|\mathbf{w}',z',v')$ as defined in \eqref{eq:notations_probs}, \eqref{eq:IF_suvfunc} and \eqref{eq:IF_z} for $(t',a',\mathbf{w}',z',v') \in {\cal T} \times \mathcal{A} \times \mathbf{W} \times {\cal Z} \times g({\cal Z})$, we define $Q(t|a',\mathbf{w}', v') = 1/[1-\sum_{z' \in \mathcal{Z}}S(t|a',\mathbf{w}',z')p(z'|\mathbf{w}',v')]$, and
\begin{align}\label{eq:IF_gamma}
  &f_{\gamma}(\cdot|t',a',\mathbf{w}',v') \colon u \mapsto Q(t'|0,\mathbf{w}',v')\sum_{z' \in \mathcal{Z}}\Big\{-p(z'|\mathbf{w}',v')\xi(u|t',a',\mathbf{w}',z')\\
  &\quad + [1-S(t'|a',\mathbf{w}', z')]h(u|\mathbf{w}',z',v')
  -\gamma(t'|a',v')\Big[ -p(z'|\mathbf{w}',v')\xi(u|t',0,\mathbf{w}',z')\nonumber \\
  &\quad + [1-S(t'|0,\mathbf{w}',z')]h(u|\mathbf{w}',z',v') \Big]\Big\}. \nonumber
\end{align}
The theoretical results for the asymptotic linearity of $(\gamma_n(t|a,v))_{a \in [k]}$ and its joint limiting distribution are provided as follows.
\begin{lemma} \label{lemma:joint_distribution_RRs}
Given \ref{assump:indep_AZ_on_W}-\ref{assump:censoring} and the stratum $v$, we establish the asymptotic linearity of the estimated relative risks of the active interventions compared to the control at time $t$, $(\gamma_n(t|a,v))_{a \in [k]} \colon$
\begin{align} \label{eq:asymp_linear_RRs_covariate_adj}
  &n^{1/2}\big(\gamma_n(t|a,v)-\gamma(t|a,v)\big)_{a \in [k]}
  = \Big(n^{-1/2}\sum_{i=1}^n f_{\gamma}(U_i|t,a,\mathbf{w}_{a},v)\Big)_{a \in [k]} + o_p(1).
  \end{align}
Moreover, $\{n^{1/2}\big(\gamma_n(t|a,v)-\gamma(t|a,v)\big)_{a \in [k]} : t \in \mathcal{T}\}$ weakly converges to a $k$-variate Gaussian process with mean zero and covariance function 
\begin{align*}
  \bs{\Sigma}_{\gamma}(t,t^{\prime}|v) = E\big[(f_{\gamma}(U|t,a,\mathbf{w}_{a},v))_{a \in [k]} \otimes (f_{\gamma}(U|t',a,\mathbf{w}_{a},v))^\top_{a \in [k]}\big].
\end{align*}
\end{lemma}
\begin{proof}[Proof of Lemma~\ref{lemma:joint_distribution_RRs}]
To develop the asymptotic distribution of $(\gamma_n(t|a,v))_{a \in [k]}$, we will establish the asymptotic linearity of
\[
  \sum_{z \in \mathcal{Z}}\big\{[1-S_n(t|a,\mathbf{w},z)]P_n(z|\mathbf{w},v)-[1-S(t|a,\mathbf{w},z)]p(z|\mathbf{w},v)\big\},
\]
with the corresponding influence function for arbitrary $(a,\mathbf{w},z) \in {\cal A} \times \mathbf{W} \times \mathcal{Z}$ at given $v$. Applying the functional delta method will then yield the result.

The first step is to show that $P_n(z|\mathbf{w},v)$ is an asymptotically linear estimator of $p(z|\mathbf{w},v)$. Note that $V_i = g(Z_i)$ for some many-to-one function $g$, $i=1,\ldots,n$. 
Let $P_n(z,\mathbf{w},v)$ and $P_n(\mathbf{w},v)$ denote the numerator and the denominator of $P_n(z|\mathbf{w},v)$; namely the empirical estimators of $p(z, \mathbf{w}, v)$ and $p(\mathbf{w}, v)$, respectively. By definition, 
\begin{align*}
  n^{1/2}\big[P_n(z,\mathbf{w},v)-p(z, \mathbf{w}, v)\big] &= n^{-1/2}\sum_{i=1}^n\big[1(Z_i=z, W_i \in \mathbf{w}, V_i=v)-p(z, \mathbf{w}, v)\big],\\
  n^{1/2}\big[P_n(\mathbf{w},v)-p(\mathbf{w}, v)\big] &= n^{-1/2}\sum_{i=1}^n\big[1(W_i \in \mathbf{w}, V_i=v)-p(\mathbf{w},v)\big].
\end{align*}
This gives that $\{n^{1/2}[P_n(z|\mathbf{w},v) - p(z|\mathbf{w},v)] : (z,\mathbf{w}) \in \mathcal{Z} \times \mathbf{W}\}=\{n^{-1/2}\sum_{i=1}^n h(U_i|\mathbf{w},z,v) : (\mathbf{w},z) \in \mathbf{W} \times \mathcal{Z}\}
\rightsquigarrow \mathcal{G}(0, \sigma^2_p)$,
where the covariance function
\begin{align*}
  \sigma^2_p\big((\mathbf{w},z),(\mathbf{w}',z')\big)=E[h(U|\mathbf{w},z,v)h(U|\mathbf{w}',z',v)] 
\end{align*}
with $h$ as defined in \eqref{eq:IF_z}.
Now we derive the influence functions of
\[
  \sum_{z \in \mathcal{Z}}\Big\{\big[1-S_n(t|a,\mathbf{w},z)\big]P_n(z|\mathbf{w},v)-\big[1-S(t|a,\mathbf{w},z)\big]p(z|\mathbf{w},v)\Big\}
\]
based on those of $P_n(z|\mathbf{w},v)$ and $S_n(t|a,\mathbf{w},z)$, for arbitrary $(a,\mathbf{w},z) \in {\cal A} \times \mathbf{W} \times {\cal Z}$.
Combining the above result with Lemma \ref{lemma:asymp_surv_functions}, we have
the below asymptotic linear approximations:
\begin{align*}
    &\left(n^{1/2}\sum_{z\in\mathcal{Z}}\left\{[1-S_n(t|0,\mathbf{w}_a,z)]P_n(z|\mathbf{w}_a,v)-[1-S(t|0,\mathbf{w}_a,z)]p(z|\mathbf{w}_a,v)\right\}\right)_{a\in [k]} \\
    &= \left(n^{-1/2}\sum_{i=1}^n\sum_{z \in \mathcal{Z}}\big\{-p(z|\mathbf{w}_{a},v)\xi(U_i|t,0,\mathbf{w}_{a},z) + \left[1-S(t|0,\mathbf{w}_{a},z)\right]h(U_i|\mathbf{w}_{a},z,v)\big\} \right)_{a \in [k]} + o_p(1)
\end{align*}
and 
\begin{align*}
    &\left(n^{1/2}\sum_{z\in\mathcal{Z}}\left\{[1-S_n(t|a,\mathbf{w}_a,z)]P_n(z|\mathbf{w}_a,v)-[1-S(t|a,\mathbf{w}_a,z)]p(z|\mathbf{w}_a,v)\right\}\right)_{a\in [k]} \\
    &= \left(n^{-1/2}\sum_{i=1}^n\sum_{z \in \mathcal{Z}}\big\{-p(z|\mathbf{w}_{a},v)\xi(U_i|t,a,\mathbf{w}_{a},z)  +\left[1-S(t|a,\mathbf{w}_{a},z)\right]h(U_i|\mathbf{w}_{a},z,v)\big\} \right)_{a \in [k]} + o_p(1).
\end{align*}
Together with the above two displays, applying the delta method and the functional central limit theorem yields the asymptotic normality of $(\gamma_n(t|a,v))_{a \in [k]}$:
\begin{align*} 
  &\Big\{n^{1/2}\big(\gamma_n(t|a,v)-\gamma(t|a,v)\big)_{a \in [k]} : t \in \mathcal{T} \Big\}
  = \Big\{\Big(n^{-1/2}\sum_{i=1}^n f_{\gamma}(U_i|t,a,\mathbf{w}_{a},v)\Big)_{a \in [k]}: t \in \mathcal{T} \Big\}
  + o_p(1)\\
  & \rightsquigarrow \mathcal{G}_k(\bs{0}_k, \bs{\Sigma}_{\gamma}(t,t^{\prime}|v)),
\end{align*}
where $f_{\gamma}$ is as defined in \eqref{eq:IF_gamma} and
$\bs{\Sigma}_{\gamma}(t,t^{\prime}|v) = E\big[(f_{\gamma}(U|t,a,\mathbf{w}_{a},v))_{a \in [k]} \otimes (f_{\gamma}(U|t',a,\mathbf{w}_{a},v))^\top_{a \in [k]}\big]$.
\end{proof}

\subsection{Preliminaries for the proof of Theorem \ref{Thm:improve_efficiency_sharing_control}}
\label{app-sec:appendix_sec_efficiency_gain}

In what follows we let $\eta(t|a,\mathbf{w},z,v) := 1-S(t|a,\mathbf{w},z)-\gamma(t|a,v)[1-S(t|0,\mathbf{w},z)]$, and for active interventions $a_1 \not= a_2$,
\begin{align} \label{eq:covariance_VE_p}
  & \sigma_{a_1,a_2}(t|v)^2 = \sum_{j=1}^2\Theta_j^2Q^2(t|0,\mathbf{w}_{a_j},v)E\bigg\{\bigg[\sum_{z \in \mathcal{Z}}
  \Big[-p(z|\mathbf{w}_{a_j},v)\xi(U|t,a_j,\mathbf{w}_{a_j},z)\\ 
  &\hspace{3cm} +\gamma(t|a_j,v)p(z|\mathbf{w}_{a_j},v)\xi(U|t,0,\mathbf{w}_{a_j},z) + \eta(t|a_j,\mathbf{w}_{a_j},z,v)h(U|\mathbf{w}_{a_j},z,v) \Big]\bigg]^2\bigg\} \nonumber\\
  &\quad + 2\prod_{j=1}^2\Big\{\Theta_jQ(t|0,\mathbf{w}_{a_j},v)\Big\}E\bigg\{\prod_{j=1}^2\bigg[\sum_{z \in \mathcal{Z}}\Big[-p(z|\mathbf{w}_{a_j},v)\xi(U|t,a_j,\mathbf{w}_{a_j},z) \nonumber\\ 
  &\hspace{3cm} +\gamma(t|a_j,v)p(z|\mathbf{w}_{a_j},v)\xi(U|t,0,\mathbf{w}_{a_j},z) +\eta(t|a_j,\mathbf{w}_{a_j},z,v)h(U|\mathbf{w}_{a_j}, z, v) \Big]\bigg]\bigg\}. \nonumber
\end{align}

We begin by developing the asymptotic normality of $\Theta\big(\gamma_{n}(t|a_1,v), \gamma_{n}(t|a_2,v)\big)$, the relative efficacy estimated based on data from a platform trial.
\begin{lemma}\label{lemma:limiting_dist_platform_trials}
Suppose that \ref{assump:indep_AZ_on_W}-\ref{assump:negative_product_firstderivatives} hold. For any given time $t$, arbitrary interventions $a_1 \not= a_2$, $n^{1/2}\big\{\Theta\big(\gamma_{n}(t|a_1,v), \gamma_{n}(t|a_2,v)\big)-\Theta\big(\gamma(t|a_1,v), \gamma(t|a_2,v)\big)\big\}$ converges in distribution to
 $\mathcal{N}(0, \sigma_{a_1,a_2}(t|v)^2)$.
\end{lemma}
\begin{proof}[Proof of Lemma~\ref{lemma:limiting_dist_platform_trials}]
The asymptotic normality is an immediate consequence of the delta method and Lemma \ref{lemma:joint_distribution_RRs}. The limiting variance of $\Theta(\gamma_n(t|a_1), \gamma_n(t|a_2))$ takes the form $$E\{[\Theta_1f_{\gamma}(U|t,a_1,\mathbf{w}_{a_1})+\Theta_2f_{\gamma}(U|t,a_2,\mathbf{w}_{a_2})]^2\},$$
and, by \eqref{eq:IF_logRR}, we have the form of $\sigma_{a_1,a_2}(t|v)^2$ as in \eqref{eq:covariance_VE_p}.
\end{proof}

We now develop the asymptotic normality of $\Theta\big(\gamma_{m}^\dagger(t|a_1,v), \gamma_{m}^\dagger(t|a_2,v)\big)$, the relative efficacy estimated based on data from separate, intervention-specific trials. When doing this, we will use
\ref{assump:identical_distribution} in several ways. First, this condition enables the limiting variance $\sigma^\dagger_{a_1,a_2}(t|v)$ to be expressed in terms of $f_{\varphi}$ as defined in \eqref{eq:IF_logRR}, with the argument $\mathbf{w}'$ equal to $\mathbf{w}_{a_1}$ or $\mathbf{w}_{a_2}$. Second, \ref{assump:identical_distribution} implies that, for all $a$,
\begin{align*}
  &P(T>t|A=a, W \in \mathbf{w}_{a}, Z=z) = \frac{P(T>t, A=a, W \in \mathbf{w}_{a}, Z=z, A\in\{0,a\})}{P(A=a, W \in \mathbf{w}_{a}, Z=z, A\in\{0,a\})}\\
  & = \frac{P(T>t, A=a, Z=z|W \in \mathbf{w}_{a}, A\in\{0,a\})}{P(A=a, Z=z | W \in \mathbf{w}_{a}, A\in\{0,a\})}
  = \frac{P_{a}^\dagger(T_a>t, A_a=a, Z_a=z)}{P_{a}^\dagger(A_a=a, Z_a=z)}\\
  & = P_{a}^\dagger(T_a>t|A_a=a, Z_a=z)
\end{align*}
and $P(T>t|A=0, W \in \mathbf{w}_{a}, Z=z) = P_{a}^\dagger(T_a>t|A_a=0, Z_a=z)$, the above displays ensure the two trials have identical conditional survival functions $S(\cdot|a,\mathbf{w}_{a},z)$ and $S(\cdot|0,\mathbf{w}_{a},z)$, given the covariate subgroup $Z=z$ and the corresponding enrollment windows in $\mathbf{w}_{a}$. Therefore, it implies that the two trials pursue the same estimands.

The following identity, which is implied by \ref{assump:identical_ratios}, will be used later:
\begin{align} \label{eq:identical_ratios_in_prob} 
  \frac{m}{m_a} = \frac{\sum_{a'=1}^k  P(A\in\{0,a'\},W\in\mathbf{w}_{a'})}{P(A\in\{0,a\},W\in\mathbf{w}_a)}.
\end{align} 
To see why the above is implied by \ref{assump:identical_ratios}, note that, under \ref{assump:identical_ratios}, $m_{a'}/m_a=n_{a'}/n_a$ for all pairs $(a,a')$. Hence, $\sum_{a'=1}^k m_{a'}/m_a = \sum_{a'=1}^k n_{a'}/n_a$. Combining this with 
the fact that $n_a := nP(A\in\{0,a\}, W\in\mathbf{w}_a)$, this gives the desired condition.

The upcoming lemma makes use of the a limiting variance $\sigma^\dagger_{a_1,a_2}(t|v)$, which is defined as follows for an arbitrary two active intervention $a_1 \not= a_2$:
\begin{align} \label{eq:covariance_VE_s}
  &\sigma^\dagger_{a_1,a_2}(t|v)^2 = \sum_{j=1}^2 \frac{m}{m_{a_j}}\Theta_j^2Q^2(t|0,\mathbf{w}_{a_j},v)E\bigg\{\bigg[\sum_{z \in \mathcal{Z}}\Big[-p(z|\mathbf{w}_{a_j},v)\xi(U|t,a_j,\mathbf{w}_{a_j},z) \\ 
  &\hspace{3cm} +\gamma(t|a_j,v)p(z|\mathbf{w}_{a_j},v)\xi(U|t,0,\mathbf{w}_{a_j},z) + \eta(t|a_j,\mathbf{w}_{a_j},z,v)h(U|\mathbf{w}_{a_j},z,v) \Big]\bigg]^2\bigg\}. \nonumber
\end{align}
We note that $\sigma^\dagger_{a_1,a_2}(t|v)^2$ does not depend on the sizes of $(m_a)_{a\in [k]}$ of the $k$ separate trials when \eqref{eq:identical_ratios_in_prob} holds since \eqref{eq:identical_ratios_in_prob} implies $m/m_{a_1}$ and $m/m_{a_2}$ are fixed numbers.

\begin{lemma}\label{lemma:limiting_dist_separate_trials}
 Suppose that \ref{assump:indep_AZ_on_W}-\ref{assump:identical_ratios} hold. For any given time $t$, arbitrary active interventions $a_1 \not= a_2$, $m^{1/2}\big\{\Theta\big(\gamma^\dagger_{m}(t|a_1,v), \gamma^\dagger_{m}(t|a_2,v)\big)-\Theta\big(\gamma(t|a_1,v), \gamma(t|a_2,v)\big)\big\}\overset{d}{\rightarrow}\mathcal{N}(0, \sigma^\dagger_{a_1,a_2}(t|v)^2)$.
\end{lemma}
\begin{proof}[Proof of Lemma~\ref{lemma:limiting_dist_separate_trials}]
Recall that, for $a \in \{a_1,a_2\}$, $P_{ma}^{\dagger}(z|v) = \sum_{i=1}^{m_{a}}1(Z_{ai}=z,V_{ai}=v)/\sum_{i=1}^{m_{a}}1(V_{ai}=v)$ and
\begin{align*}
  \gamma_{m}^\dagger(t|a,v) = \frac{1-\sum_{z \in \mathcal{Z}}S^\dagger_m(t|a,z)P^\dagger_{ma}(z|v)}{1-\sum_{z \in \mathcal{Z}}S^\dagger_m(t|0_a,z)P^\dagger_{ma}(z|v)},
\end{align*}
where $S^\dagger_m(t|a,z)$ and $S^\dagger_m(t|0_a,z)$ are the Kaplan-Meier estimators of $P(T>t)$, respectively, within the strata $(A=a, Z=z)$ and $(A=0_a, Z=z)$ based on data from the separate trial for intervention $a$.
Following similar arguments as were used to establish Lemma~\ref{lemma:joint_distribution_RRs}, it is possible to show that
\begin{align*}
  &\begin{pmatrix}
  m_{a_1}^{1/2}\big[\gamma_{m}^\dagger(t|a_1,v)-\gamma(t|a_1,v)\big]\\
  m_{a_2}^{1/2}\big[\gamma_{m}^\dagger(t|a_2,v)-\gamma(t|a_2,v)\big]
  \end{pmatrix}
  = \begin{pmatrix}
  m_{a_1}^{-1/2}\sum_{i=1}^{m_{a_1}} f_{\gamma}(U_{a_1i}|t,a_1,\mathbf{w}_{a_1},v)\\
  m_{a_2}^{-1/2}\sum_{i=1}^{m_{a_2}} f_{\gamma}(U_{a_2i}|t,a_2,\mathbf{w}_{a_2},v)
  \end{pmatrix}
  + o_p(1), 
\end{align*}
where $f_{\gamma}$ is defined in \eqref{eq:IF_gamma}. These arguments are omitted for brevity. By the delta method and the central limit theorem, the above implies that
\begin{align*}
  &\sqrt{m}\big\{\Theta\big(\gamma_{m}^\dagger(t|a_1,v), \gamma_{m}^\dagger(t|a_2,v)\big)-\Theta\big(\gamma(t|a_1,v), \gamma(t|a_2,v)\big)\big\}\\
  &= \frac{\sqrt{m}}{m_{a_1}}\sum_{i=1}^{m_{a_1}}\Theta_1f_{\gamma}(U_{a_1i}|t,a_1,\mathbf{w}_{a_1},v)
  + \frac{\sqrt{m}}{m_{a_2}}\sum_{i=1}^{m_{a_2}}\Theta_2f_{\gamma}(U_{a_2i}|t,a_2,\mathbf{w}_{a_2},v) + o_p(1)
  \rightsquigarrow \mathcal{N}(0, \sigma^\dagger_{a_1,a_2}(t|v)^2).
\end{align*}
\end{proof}

The following lemma makes use of the function $h$, which is defined in \eqref{eq:IF_z}.
\begin{lemma} \label{lemma:property_h}
For any stratum $v$, covariate levels $z$ and $\tilde{z}$ such that $g(z)=v$ and $g(\tilde{z})=v$, and active interventions $a$ and $a'$,
$$E\big[h(U|\mathbf{w}_{a}, z, v)h(U|\mathbf{w}_{a'}, \tilde{z}, v)\big] = K(a,a',\tilde{z})p(z|\mathbf{w}_a,v) + K'(a,a',\tilde{z})p(z|\mathbf{w}_a \cap \mathbf{w}_{a'},v),$$ 
where 
$K(a,a',\tilde{z})$ and $K'(a,a',\tilde{z})$ are constants that depend on $(a,a')$ and $\tilde{z}$ but do not depend on $z$.
\end{lemma}
\begin{proof}[Proof of Lemma~\ref{lemma:property_h}]
The result follows immediately from the fact that
\begin{align*}
  &E\big[h(U|\mathbf{w}_{a}, z, v)h(U|\mathbf{w}_{a'}, \tilde{z}, v)\big]
  = \bigg[ p(\mathbf{w}_{a} \cap \mathbf{w}_{a'}, v)\frac{p(\tilde{z}|\mathbf{w}_{a'},v)}{p(\mathbf{w}_{a},v)p(\mathbf{w}_{a'},v)}-\frac{p(\tilde{z},\mathbf{w}_{a} \cap \mathbf{w}_{a'},v)}{p(\mathbf{w}_{a},v)p(\mathbf{w}_{a'},v)}\bigg]p(z|\mathbf{w}_{a},v)\\
  &\quad + \bigg[\frac{-p(\tilde{z}|\mathbf{w}_{a'},v)}{p(\mathbf{w}_{a},v)p(\mathbf{w}_{a'},v)}p(\mathbf{w}_{a} \cap \mathbf{w}_{a'},v)\bigg]p(z|\mathbf{w}_{a} \cap \mathbf{w}_{a'},v)\\
  & \equiv K(a,a',\tilde{z})p(z|\mathbf{w}_a,v) + K'(a,a',\tilde{z})p(z|\mathbf{w}_a \cap \mathbf{w}_{a'},v).
\end{align*}
\end{proof}

\subsection{Proof of Theorem~\ref{Thm:improve_efficiency_sharing_control}}
Recalling the definition of $\sigma^\dagger_{a_1, a_2}(t|v)^2$ from \eqref{eq:covariance_VE_s} and multiplying each term $j\in \{1,2\}$ in that expression by $P(A \in \{0, a_j\}, W \in \mathbf{w}_{a_j})$ gives that
\begin{align} \label{eq:covariance_VE_s_lb} 
  &\sigma^\dagger_{a_1, a_2}(t|v)^2 \ge \sum_{j=1}^2\frac{m}{m_{a_j}}P(A \in \{0, a_j\}, W \in \mathbf{w}_{a_j}) \Theta_j^2Q^2(t|0,\mathbf{w}_{a_j},v) E\bigg\{\bigg[\sum_{z \in \mathcal{Z}}\Big[-p(z|\mathbf{w}_{a_j},v)\\
  &\quad \times \xi(U|t,a_j,\mathbf{w}_{a_j},z)
  +\gamma(t|a_j,v)p(z|\mathbf{w}_{a_j},v)\xi(U|t,0,\mathbf{w}_{a_j},z) + \eta(t|a_j,\mathbf{w}_{a_j},z,v)h(U|\mathbf{w}_{a_j},z,v) \Big]\bigg]^2\bigg\}. \nonumber
\end{align}  
Plugging the identity for $m/m_{a_j}$ from \eqref{eq:identical_ratios_in_prob} into the right-hand side above gives that
\begin{align*}  
  & \sigma^\dagger_{a_1, a_2}(t|v)^2 \ge \bigg[\sum_{a=1}^k P(A\in\{0,a\}, W\in \mathbf{w}_{a})\bigg]\sum_{j=1}^2\Theta_j^2Q^2(t|0,\mathbf{w}_{a_j},v)E\bigg\{\bigg[
  \sum_{z \in \mathcal{Z}}\Big[-p(z|\mathbf{w}_{a_j},v) \\
  &\quad \times \xi(U|t,a_j,\mathbf{w}_{a_j},z)
  +\gamma(t|a_j,v)p(z|\mathbf{w}_{a_j},v)\xi(U|t,0,\mathbf{w}_{a_j},z) +\eta(t|a_j,\mathbf{w}_{a_j},z,v) h(U|\mathbf{w}_{a_j},z,v) \Big]\bigg]^2\bigg\}. \nonumber
\end{align*}
Shortly, we will show that the third term of $\sigma_{a_1,a_2}(t|v)^2$ given in \eqref{eq:covariance_VE_p} is nonpositive. Before we provide the somewhat lengthy arguments required to do this, we show that establishing that this term is nonpositive will give the desired result. Once we show that that term is nonpositive, the above will give that
\begin{align*}
  \frac{\sigma^\dagger_{a_1,a_2}(t|v)^2}{\sigma_{a_1,a_2}(t|v)^2} \ge \sum_{a'=1}^k P(A\in\{0,a'\}, W\in \mathbf{w}_{a'}).   
\end{align*}
Also, because $\sigma_n$ and $\sigma^\dagger_m$ are consistent estimators of $\sigma_{a_1,a_2}(t|v)$ and $\sigma^\dagger_{a_1,a_2}(t|v)$ and $n/m\rightarrow \kappa$, we have that
\begin{align*}
    \left|\frac{\omega_n}{\omega_m^\dagger}- \Big[\frac{\sigma^\dagger_{a_1,a_2}(t|v)}{m^{1/2}}\Big]^{-1}\frac{\sigma_{a_1,a_2}(t|v)}{n^{1/2}}\right|&\overset{p}{\rightarrow} 0.
\end{align*}
Using that $n/m\rightarrow \kappa$, this shows that $\omega_n^2/(\omega_m^\dagger)^2\overset{p}{\rightarrow} \kappa^{-1} \sigma_{a_1,a_2}(t|v)^2/\sigma^\dagger_{a_1,a_2}(t|v)^2\le 1/\kappa\sum_{a'=1}^k P(A\in\{0,a'\}, W\in \mathbf{w}_{a'})=\rho/\kappa$. Hence, ${\rm plim}_{m,n}\,\omega_n/\omega_m^\dagger \le 1$ if $\kappa\ge \rho$, and this inequality is strict if $\kappa>\rho$.

The above argument was predicated on showing that the last term in \eqref{eq:covariance_VE_p} is nonpositive. We will do this in what follows. The expected cross-product in this last term can be decomposed as  
\begin{align} \label{eq:decomp_expect_crossprod}
  &E\bigg\{\sum_{z \in \mathcal{Z}}\prod_{j=1}^2\Big[-p(z|\mathbf{w}_{a_j},v)\xi(U|t,a_j,\mathbf{w}_{a_j},z)
  +\gamma(t|a_j,v)p(z|\mathbf{w}_{a_j},v)\xi(U|t,0,\mathbf{w}_{a_j},z)\\
  & \hspace{1.8cm} + \eta(t|a_j,\mathbf{w}_{a_j},z,v) h(U|\mathbf{w}_{a_j},z,v) \Big]\bigg\} \nonumber \\
  & + E\bigg\{\sum_{z \in \mathcal{Z}}\Big\{\Big[-p(z|\mathbf{w}_{a_1},v)\xi(U|t,a_1,\mathbf{w}_{a_1},z)
  +\gamma(t|a_1,v)p(z|\mathbf{w}_{a_1},v)\xi(U|t,0,\mathbf{w}_{a_1},z) \nonumber \\
  & \hspace{1.8cm} + \eta(t|a_1,\mathbf{w}_{a_1},z,v) h(U|\mathbf{w}_{a_1},z,v) \Big] 
  \sum_{\tilde{z} \in \mathcal{Z}\backslash\{z\}}\Big[-p(\tilde{z}|\mathbf{w}_{a_2},v)\xi(U|t,a_2,\mathbf{w}_{a_2},\tilde{z}) \nonumber\\ 
  &\hspace{1.8cm} +\gamma(t|a_2,v)p(\tilde{z}|\mathbf{w}_{a_2},v)\xi(U|t,0,\mathbf{w}_{a_2},\tilde{z}) + \eta(t|a_2,\mathbf{w}_{a_2},\tilde{z},v) h(U|\mathbf{w}_{a_2},\tilde{z},v) \Big]\Big\}\bigg\}. \nonumber
\end{align} 
Because $\xi(U|t,0,\mathbf{w}_{a_j},z)$ is zero if $Z\not=z$, and similarly for $\xi(U|t,0,\mathbf{w}_{a_j},\tilde{z})$, the fact that $z\not=\tilde{z}$ implies that $\xi(U|t,0,\mathbf{w}_{a_1},z)\xi(U|t,0,\mathbf{w}_{a_2},\tilde{z})=0$. Moreover, by the definitions of $h$ in \eqref{eq:IF_z}, 
\begin{align*}
  &E\Big\{\xi(U|t,0,\mathbf{w}_{a_j},z)h(U|\mathbf{w}_{a_k},\tilde{z},v)\Big\}\\
  &= E\Big\{\xi(U|t,0,\mathbf{w}_{a_j},z)\frac{1(W \in \mathbf{w}_{a_k}, V=v)}{p(\mathbf{w}_{a_k},v)}\left[1(Z=\tilde{z}) - p(\tilde{z}|\mathbf{w}_{a_k},v)\right]\Big\}\\
  & = -\frac{p(\tilde{z}|\mathbf{w}_{a_k},v)}{p(\mathbf{w}_{a_k},v)}E\Big\{\xi(U|t,0,\mathbf{w}_{a_j} \cap \mathbf{w}_{a_k},z)\Big\}=0,
\end{align*}
where the second equality holds because $\xi(U|t,0,\mathbf{w}_{a_j},z)=0$ whenever $Z\not=z$ and the third equality follows by direct calculation. Similarly,
$E\big\{\xi(U|t,0,\mathbf{w}_{a_j},\tilde{z})h(U|\mathbf{w}_{a_k},z,v)\big\}=0$.

Returning to \eqref{eq:decomp_expect_crossprod}, we see that the 
the second term in that display can be simplified as
\begin{align*}  
  \sum_{z \in \mathcal{Z}}\eta(t|a_1,\mathbf{w}_{a_1},z,v)\bigg\{\sum_{\tilde{z} \in \mathcal{Z}\backslash\{z\}}
  \eta(t|a_2,\mathbf{w}_{a_2},\tilde{z},v) E\big[h(U|\mathbf{w}_{a_1},z,v)h(U|\mathbf{w}_{a_2},\tilde{z},v)\big]\bigg\},
\end{align*}  
and by Lemma \ref{lemma:property_h}, written as
\begin{align} \label{eq:reduced_term}
  & \sum_{z \in \mathcal{Z}}\eta(t|a_1,\mathbf{w}_{a_1},z,v)\bigg\{\sum_{\tilde{z} \in \mathcal{Z}}
  \eta(t|a_2,\mathbf{w}_{a_2},\tilde{z},v)\big[K(a_1,a_2,\tilde{z})p(z|\mathbf{w}_{a_1},v) + K'(a_1,a_2,\tilde{z})p(z|\mathbf{w}_{a_1} \cap \mathbf{w}_{a_2},v)\big]\bigg\} \nonumber \\
  & \quad - \sum_{z \in \mathcal{Z}}\eta(t|a_1,\mathbf{w}_{a_1},z,v)\eta(t|a_2,\mathbf{w}_{a_2},z,v)E\big[h(U|\mathbf{w}_{a_1},z,v)h(U|\mathbf{w}_{a_2},z,v)\big].
\end{align}
The first line of \eqref{eq:reduced_term} above is zero, because the constancy assumption \ref{assump:constancy} implies that 
\begin{align*}
  \sum_{z}\eta(t|a_1,\mathbf{w}_{a_1},z,v)p(z|\mathbf{w}_{a_1},v) = 0 \;\mbox{ and }\;  \sum_{z}\eta(t|a_1,\mathbf{w}_{a_1},z,v)p(z|\mathbf{w}_{a_1} \cap \mathbf{w}_{a_2},v) = 0.  
\end{align*}
Let $K = 2\big[\prod_{j=1}^2\Theta_jQ(t|0,w_{a_j},v)\big]$.
Plugging the remaining quantity in \eqref{eq:reduced_term} into the decomposition in \eqref{eq:decomp_expect_crossprod}, the last term of $\sigma_{a_1,a_2}(t|v)^2$ in \eqref{eq:covariance_VE_p} reduces to
\begin{align}\label{eq:last_term_of_sigma}
  K\sum_{z} &\bigg\{ - \sum_{j,k\in [2]:j \not= k}\bigg[
  p(z|\mathbf{w}_{a_{k}},v)\eta(t|a_j,\mathbf{w}_{a_j},z,v) E\Big[\xi(U|t,a_{k},\mathbf{w}_{a_{k}},z)h(U|\mathbf{w}_{a_j},z,v)\Big] \bigg] \nonumber \\
  & + \sum_{j,k\in [2]:j \not= k}\bigg[\gamma(t|a_{k},v)p(z|\mathbf{w}_{a_{k}},v)\eta(t|a_j,\mathbf{w}_{a_j},v) E\Big[\xi(U|t,0,\mathbf{w}_{a_k},z)h(U|\mathbf{w}_{a_j},z,v) \Big]\bigg] \nonumber\\
  & + \bigg[\prod_{j=1}^2\gamma(t|a_j,v)p(z|\mathbf{w}_{a_j},v)\bigg]E\bigg[\prod_{j=1}^2\xi(U|t,0,\mathbf{w}_{a_j},z)\bigg] \bigg\}.
\end{align}
Our objective is to show that \eqref{eq:last_term_of_sigma} is nonpositive.
For each $z\in\mathcal{Z}$ and $(j,k) \in [2], j \not= k$, we first show that
\begin{align} \label{eq:zero_expectations}
E\big[\xi(U|t,a_k,\mathbf{w}_{a_k},z)h(U|\mathbf{w}_{a_j},z,v)\big] = 0;\:\: E\big[\xi(U|t,0,\mathbf{w}_{a_k},z)h(U|\mathbf{w}_{a_j},z,v)\big] = 0.
\end{align}
From the definition of $h$ in \eqref{eq:IF_z},
we can easily see that
$$E\big[\xi(U|t,a_k,\mathbf{w}_{a_k},z)h(U|\mathbf{w}_{a_j},z,v)\big] = S(t|a_k,\mathbf{w}_{a_k},z)\frac{1-p(z|\mathbf{w}_{a_j},v)}{p(\mathbf{w}_{a_j},v)}E\big[\xi(U|t,a_k,\mathbf{w}_{a_j} \cap \mathbf{w}_{a_k},z)\big] = 0$$
and similarly $E\big[\xi(U|t,0,\mathbf{w}_{a_k},z)h(U|\mathbf{w}_{a_j},z,v)\big] = 0$, so we prove \eqref{eq:zero_expectations}.

By \ref{assump:negative_product_firstderivatives}, we know
that $\Theta_1\Theta_2 <0$; along with each of $Q(\cdot)$, $\gamma(\cdot)$, $S(\cdot)$ and $p(\cdot)$ being non-negative,
$K$ in \eqref{eq:last_term_of_sigma} is non-positive.
To show the desired result, it therefore suffices to show that $E\big[\xi(U|t,0,\mathbf{w}_{a_1},z)\xi(U|t,0,\mathbf{w}_{a_2},z)\big] \geq 0$, for each $z$.
From \eqref{eq:IF_suvfunc}, define the stratified basic counting process and the risk process for an arbitrary participant on the control arm at time $v$ and within the covariate subgroup $Z=z$, respectively:
\begin{align*}
\tilde{N}(s;0,\mathbf{w},z)&=1(X \leq s, \Delta=1, A=0, W \in \mathbf{w}, Z=z);\\
\tilde{Y}(s;0,\mathbf{w},z)&=1(X \geq s, A=0, W \in \mathbf{w}, Z=z),
\end{align*}
and the martingale
\begin{align*}
  \tilde{M}(ds;0,\mathbf{w},z)=\tilde{N}(ds;0,\mathbf{w},z)-\tilde{Y}(s;0,\mathbf{w},z)\frac{p_1(ds,0,\mathbf{w},z)}{p_2(s,0,\mathbf{w},z)}
\end{align*}
with respect to the filtration
$\tilde{\mathcal{F}}_{s} = \sigma\{\tilde{N}(s';\cdot), \tilde{Y}(s';\cdot),\, A,\, W,\, Z: s' \leq s \in \mathcal{T}\}$. 
We have that
\begin{align*}
  & E\big[\xi(U|t,0,\mathbf{w}_{a_1},z)\xi(U|t,0,\mathbf{w}_{a_2},z)\big]\\
  & = \prod_{j=1}^2 \big\{S(t|0,\mathbf{w}_{a_j},z)\big\}E\Big[\int_0^t\frac{1}{p_2(s,0,\mathbf{w}_{a_1},z)}\tilde{M}(ds;0,\mathbf{w}_{a_1},z)\int_0^t\frac{1}{p_2(s,0,\mathbf{w}_{a_2},z)}\tilde{M}(ds;0,\mathbf{w}_{a_2},z) \Big]\\
  & = \prod_{j=1}^2 \big\{S(t|0,\mathbf{w}_{a_j},z)\big\}\int_0^t\frac{P(X \geq s, A=0, W \in \mathbf{w}_{a_1} \cap \mathbf{w}_{a_2}, Z=z)}{p_2(s,0,\mathbf{w}_{a_1},z)p_2(s,0,\mathbf{w}_{a_2},z)}\frac{p_1(ds,0,\mathbf{w}_{a_1} \cap \mathbf{w}_{a_2},z)}{p_2(s,0,\mathbf{w}_{a_1} \cap \mathbf{w}_{a_2},z)}.
\end{align*}
Together with $p_2>0$ and $p_1(ds,\cdot)/p_2(s,\cdot) > 0$, this quantity is positive
when interventions $a_1$ and $a_2$ are under randomization in some common windows, namely, $\mathbf{w}_{a_1} \cap \mathbf{w}_{a_2} \not= \varnothing$, and zero otherwise. This completes the proof.


\section{Conditions on relative efficacy measures to be used in superiority or noninferiority analyses}\label{app:partialOrder}


When discussing the relative efficacy measure $\Theta$, we mostly have in mind scenarios where $a_1$ is declared superior to $a_2$ when $\theta:=\Theta(\gamma(t|a_1,v),\gamma(t|a_2,v)))<0$ and is declared noninferior when $\theta< \varepsilon$ for some specified noninferiority margin $\varepsilon>0$. When the event with occurrence time $T$ is undesirable (e.g., death), a natural choice of $\Theta$ in such scenarios is the additive contrast $\Theta(r_1,r_2)=r_1-r_2$. When the event is desirable (e.g., recovery), $\Theta(r_1,r_2)=r_2-r_1$ could instead be used. Henceforth we assume that the event is harmful so that smaller relative risks are considered preferable.

In the above described superiority and noninferiority analyses, it is natural to require that $\Theta$ should respect a partial order on the relative risk product space $(0,\infty)^2$. This partial order makes it possible to compare two (intervention $a_1$, intervention $a_2$) relative risk pairs, denoted by $(r_1, r_2)$ and $(s_1, s_2)$,
and identify whether there is a pair where intervention $a_1$ looks indisputably more favorable as compared to intervention $a_2$. Specifically, this is the case for $(r_1,r_2)$ if at least one of the following conditions holds:
\begin{align}
  \mbox{(i)}\; r_1< s_1\;\mbox{and}\;r_2\ge s_2\,;\: 
  \mbox{(ii)}\;r_1\le s_1\;\mbox{and}\;r_2>s_2. 
                                \label{eq:partial_ordering_conditions}
\end{align}
If (i) or (ii) holds, then we say that $(r_1,r_2)\prec (s_1,s_2)$. Stated in words, $(r_1,r_2)\prec (s_1,s_2)$ indicates that intervention $a_1$ in the $(r_1,r_2)$ scenario has relative risk that is at least as low as in the $(s_1,s_2)$ scenario, intervention $a_2$ has relative risk at least as high, and at least one of these two inequalities is strict. Therefore, if $(r_1,r_2)\prec (s_1,s_2)$ and $a_1$ is declared superior (noninferior) under $(s_1,s_2)$, it is natural to require that the same declaration should be made under $(r_1,r_2)$. This suggests that $\Theta$ should be such that $\Theta(r_1,r_2)\le \Theta(s_1,s_2)$ whenever $(r_1,r_2)\prec (s_1,s_1)$. Moreover, because $\Theta(s_1,s_2)$ is used as a continuous measure of how favorable $a_1$ looks as compared to $a_2$, it is natural to ask for something slightly stronger, namely that $\Theta(r_1,r_2)< \Theta(s_1,s_2)$ whenever $(r_1,r_2)\prec (s_1,s_1)$. This condition states that the relative efficacy of $a_1$ versus $a_2$ necessarily improves if $a_1$ looks strictly more favorable as compared to $a_2$ in one scenario than another. If $\Theta(r_1,r_2)< \Theta(s_1,s_2)$ whenever $(r_1,r_2)\prec (s_1,s_1)$, then we say that $\Theta$ respects the partial order.

Many natural relative efficacy measures respect the partial order. For example, the additive contrast $\Theta(r_1,r_2)=r_1-r_2$ and the multiplicative contrast $\Theta(r_1,r_2)=r_1/r_2$ both satisfy this condition. Monotonically increasing transformations of these contrasts, such as $\Theta(r_1,r_2)=\log(r_1/r_2)$, also respect the partial order.

The below result shows that, when a differentiable relative efficacy measure $\Theta$ respects the partial order, it must be the case that \ref{assump:negative_product_firstderivatives} holds. Hence, our arguments that $\Theta$ should satisfy the partial order imply that $\Theta$ should satisfy \ref{assump:negative_product_firstderivatives}.
\begin{lemma}\label{lemma:partial_ordering}
Suppose that $\Theta$ is differentiable and that $\Theta$ respects the partial order $\prec$ in the sense that $\Theta(r_1,r_2)< \Theta(s_1,s_2)$ whenever $(r_1, r_2) \prec (s_1, s_2)$. Under these conditions, \ref{assump:negative_product_firstderivatives} holds.
\end{lemma}
\begin{proof}[Proof of Lemma~\ref{lemma:partial_ordering}]
The first part of \ref{assump:negative_product_firstderivatives}, namely that $\Theta$ is differentiable, holds by the assumption. Hence, it remains to show the latter part, namely that $\left[\frac{\partial}{\partial r_1}\Theta(r_1,r_2)\right]\left[\frac{\partial}{\partial r_2}\Theta(r_1,r_2)\right]<0$ for all $(r_1,r_2)\in (0,\infty)^2$. 

Fix $(r_1,r_2)\in (0,\infty)^2$. By the definition of the partial derivative and the fact that $(r_1, r_2) \prec (s_1, r_2)$ whenever $s_1>r_1$,
\begin{align*}
    \frac{\partial}{\partial r_1}\Theta(r_1,r_2)&= \lim_{s_1\downarrow r_1}\frac{\Theta(r_1,r_2)-\Theta(s_1,r_2)}{r_1-s_1} > 0.
\end{align*}
Similarly, since $(r_1, r_2) \prec (r_1, s_2)$ whenever $s_2<r_2$,
\begin{align*}
    \frac{\partial}{\partial r_2}\Theta(r_1,r_2)&= \lim_{s_2\uparrow r_2}\frac{\Theta(r_1,r_2)-\Theta(r_1,s_2)}{r_2-s_2} < 0.
\end{align*}
Hence, the latter part of \ref{assump:negative_product_firstderivatives} holds.
\end{proof}

\section{Likelihood ratio test} \label{app-sec:LRT}
We consider an asymptotic variant of the likelihood ratio test for establishing the noninferiority of intervention 1 as compared to the most efficacious of the other interventions, instead of the intersection test introduced in Section \ref{sec:intersection_test}. The test that we propose builds on the fact that
\begin{align}
\sqrt{n}\Sigma_n^{-1/2}(\boldsymbol{\gamma}_n-\boldsymbol{\gamma})\rightsquigarrow \mathcal{N}(0,{\rm Id}), \label{eq:gammanWeakConv}
\end{align}
where $\boldsymbol{\gamma}:= (\gamma(t|a,v))_{a=1}^k$, $\boldsymbol{\gamma}_n:= (\gamma_n(t|a,v))_{a=1}^k$, and $\Sigma_n$ is a consistent estimate of the asymptotic covariance matrix of $\boldsymbol{\gamma}_n$. In particular, when defining the form of the proposed test, we make the working assumption that the above distributional result is exact when $n$ is finite, that is, we suppose that
\begin{align}
\boldsymbol{\gamma}_n\sim \mathcal{N}(\boldsymbol{\gamma},\Sigma_n/n). \label{eq:gammanExactFinite}
\end{align}
We then show that, in fact, the proposed test provides asymptotically valid type I error control even when the above does not hold but \eqref{eq:gammanWeakConv} does.


Let $\Gamma_0:= \{\boldsymbol{\gamma}'=(\gamma_a')_{a=1}^k\in [0,\infty)^k : \gamma_1'\ge \min\{\delta,\min_{a\not=1} \gamma_a'+\epsilon\}\}$. Under \eqref{eq:gammanExactFinite} and treating $\Sigma_n$ as fixed and known, the likelihood ratio test statistic for a test of $H_0 : \boldsymbol{\gamma}\in \Gamma_0$ versus the complementary alternative takes the form
\begin{align*}
  T_n(\delta, \epsilon) 
  &= -2\log\frac{\sup_{\boldsymbol{\gamma}'\in \Gamma_0} \phi_{\boldsymbol{\gamma}',\Sigma_n/n}\left(\boldsymbol{\gamma}_n\right)}{\sup_{\boldsymbol{\gamma}'\in [0,\infty)^k} \phi_{\boldsymbol{\gamma}',\Sigma_n/n}\left(\boldsymbol{\gamma}_n\right)}= -2\log\frac{\sup_{\boldsymbol{\gamma}'\in \Gamma_0} \phi_{\boldsymbol{\gamma}',\Sigma_n/n}\left(\boldsymbol{\gamma}_n\right)}{\phi_{\boldsymbol{\gamma}_n,\Sigma_n/n}\left(\boldsymbol{\gamma}_n\right)},
\end{align*}
where $\phi_{\boldsymbol{\gamma}',\Sigma_n/n}(\boldsymbol{\gamma}_n)$ denotes the density function of a $k$-variate normal random variable with mean $\boldsymbol{\gamma}'$ and covariance matrix $\Sigma_n/n$, evaluated at $\boldsymbol{\gamma}_n$. The likelihood ratio test rejects when $T_n(\delta, \epsilon)$ exceeds the $1-\alpha$ quantile of a $\chi^2(k)$ distribution, which we denote by $q_{1-\alpha}$. Though \eqref{eq:gammanExactFinite} does not necessarily hold, this test will still maintain asymptotic type I error. To see why, note that \eqref{eq:gammanWeakConv}, together with the continuous mapping theorem, imply that $(\boldsymbol{\gamma}_n-\boldsymbol{\gamma})^\top\big(\Sigma_n/n\big)^{-1}(\boldsymbol{\gamma}_n-\boldsymbol{\gamma})\rightsquigarrow \chi^2(k)$. Combining this with the fact that under the null,
\begin{align*}
    T_n(\delta, \epsilon) &\le -2\log\frac{\phi_{\boldsymbol{\gamma},\Sigma_n/n}\left(\boldsymbol{\gamma}_n\right)}{\phi_{\boldsymbol{\gamma}_n,\Sigma_n/n}\left(\boldsymbol{\gamma}_n\right)} = (\boldsymbol{\gamma}_n-\boldsymbol{\gamma})^\top\big(\Sigma_n/n\big)^{-1}(\boldsymbol{\gamma}_n-\boldsymbol{\gamma}), 
\end{align*}
we find that $\limsup_n P\{T_n(\delta,\epsilon)\ge q_{1-\alpha}\}\le \alpha$, showing proper type I error control.

In our simulation, we use the R package \texttt{nloptr} \citep{NLoptR2020} with method=\texttt{cobyla} to solve the constrained optimization problem in the numerator of $T_n(\delta, \epsilon)$, for given $(\delta, \epsilon)$.

\newpage
\section{Additional simulation results} \label{app-sec:supp_figures}
\begin{figure}[hbt!]
\begin{center}
 \includegraphics[trim={0cm 0.1cm 0.2cm 0.1cm}, clip, width=0.8\textwidth]
 {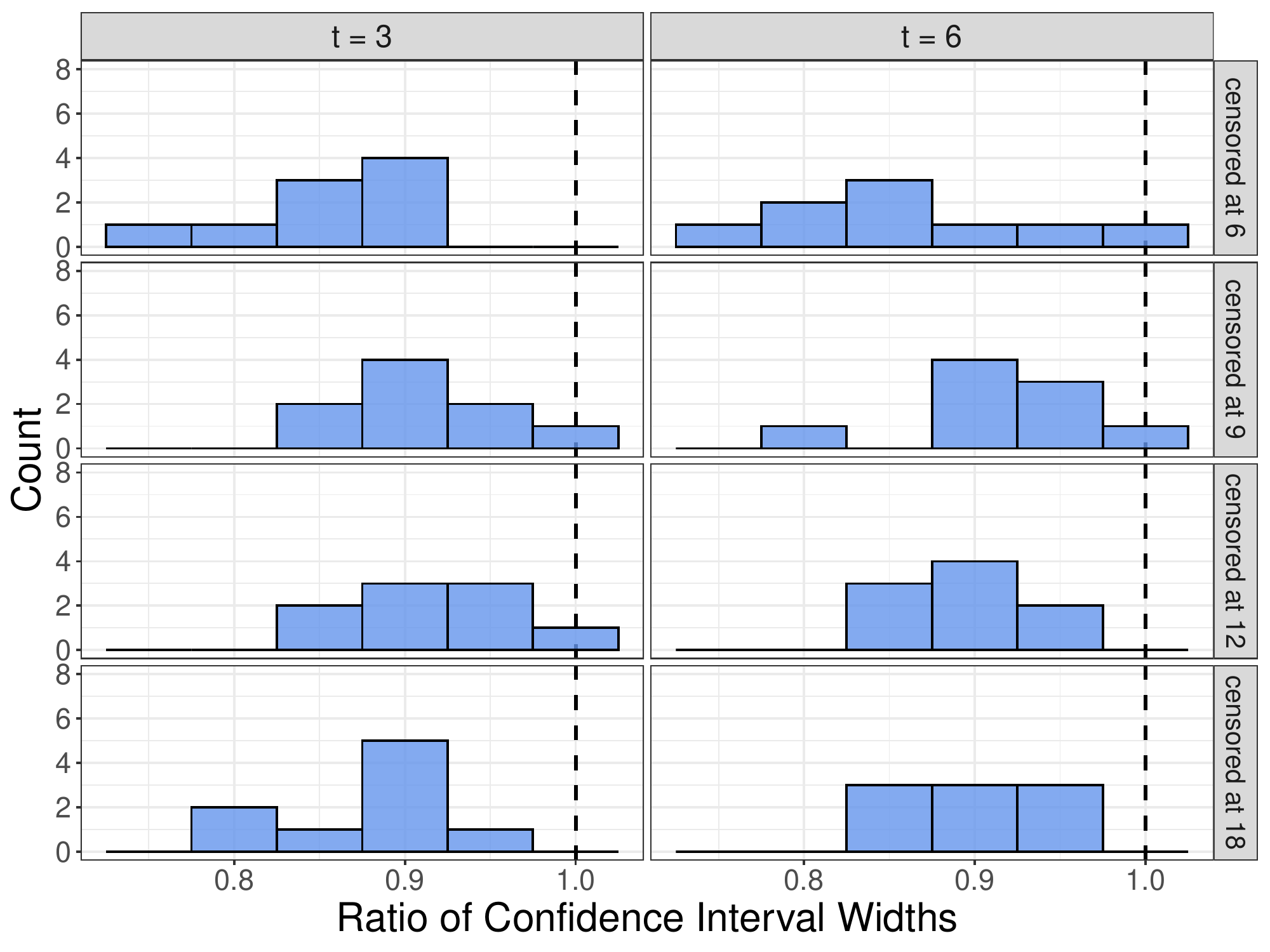}
 \caption{The confidence interval width ratios of using data from the platform trial versus using those from separate independent trials for conditional relative risk ratios, evaluated at $t=3, 6$ and subject to different administrative censoring times.}
 \label{fig:hist_cRRR_ciwidth_ratios}
\end{center}
\end{figure}

\begin{figure}[hbt!]
\begin{center}
 \includegraphics[trim={0cm 0.1cm 0.2cm 0.1cm}, clip, width=0.8\textwidth]
 {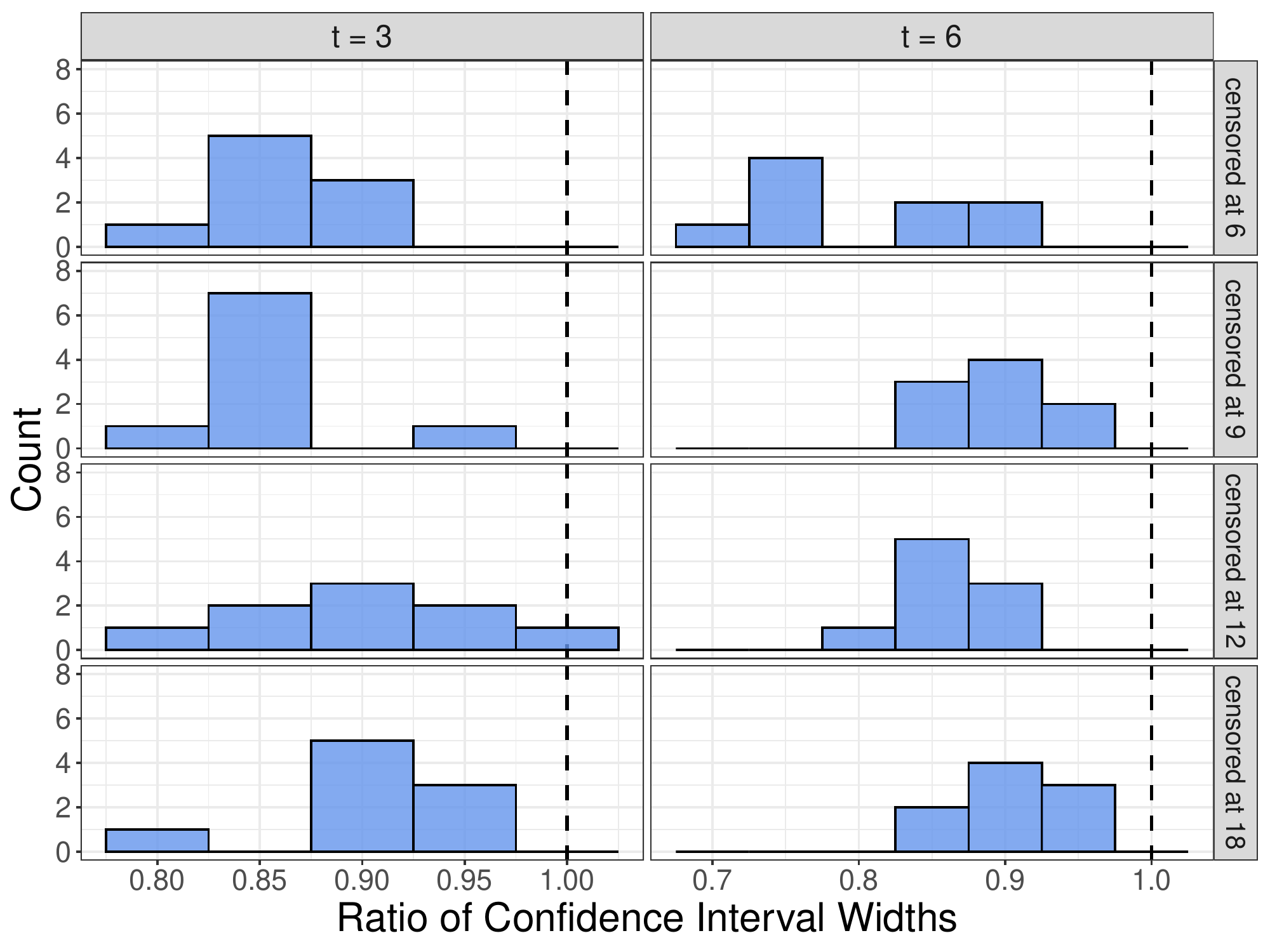}
 \caption{As in Figure \ref{fig:hist_cRRR_ciwidth_ratios}, except for marginal relative risk ratios.}
 \label{fig:hist_mRRR_ciwidth_ratios}
\end{center}
\end{figure}
\end{appendix}

\newpage
\section{Further details on data illustration}\label{app:dataIllustration}
We resampled from HVTN 703/704 data without replacement to create synthetic platform trial datasets. To do this, active intervention arms were withdrawn from randomization in some windows, so that the low dose group was under randomization in an arbitrary three windows and the high dose group in another three. 
More concretely, there were four types of windows. The first type of window, termed L, only had the low dose arm under randomization. The second, termed H, only had the high-dose arm. The third and fourth, termed B$_{\textnormal{all}}$ and B$_{\textnormal{sub}}$, had both the low- and high-dose arms under randomization. 
All low-dose, high-dose, and control arm participants enrolled in B$_{\textnormal{all}}$ were enrolled into the platform trial. A subset of participants was enrolled from the other windows to achieve a desired proportion of controls shared (25-50\%), which is defined as the number of controls that are shared by the two arms divided by the total number of controls enrolled in the platform trial. In particular, we varied the number of low-dose participants enrolled in windows L and B$_{\textnormal{sub}}$ and high-dose participants enrolled in windows H and B$_{\textnormal{sub}}$ to achieve the desired total number of participants on each active arm and proportion of controls shared. For each proportion of controls shared, the allocation ratio in H and L were fixed at 1:1 and in B$_{\textnormal{sub}}$ was fixed at 1:1:1. 
We used the data from each trial to generated 12,000 total platform trial datasets, where 500 datasets were generated for each of the 24 possible assignments of the four window types (L, H, B$_{\textnormal{all}}$, B$_{\textnormal{sub}}$) to the four chronologically defined windows mentioned earlier. 
The rationale for averaging across these 24 scenarios is that the number of events differs across the 24 possible allocations of (L, H, B$_{\textnormal{all}}$, B$_{\textnormal{sub}}$), and so otherwise misleadingly favorable (or unfavorable) results for the platform trial can be observed, depending on whether B$_{\textnormal{sub}}$ has more (or fewer) events than do L and H.

\begin{figure}[hbt!]
\begin{center}
 \includegraphics[trim={0.1cm 0.1cm 0.1cm 0.7cm}, clip, width=0.7\textwidth]{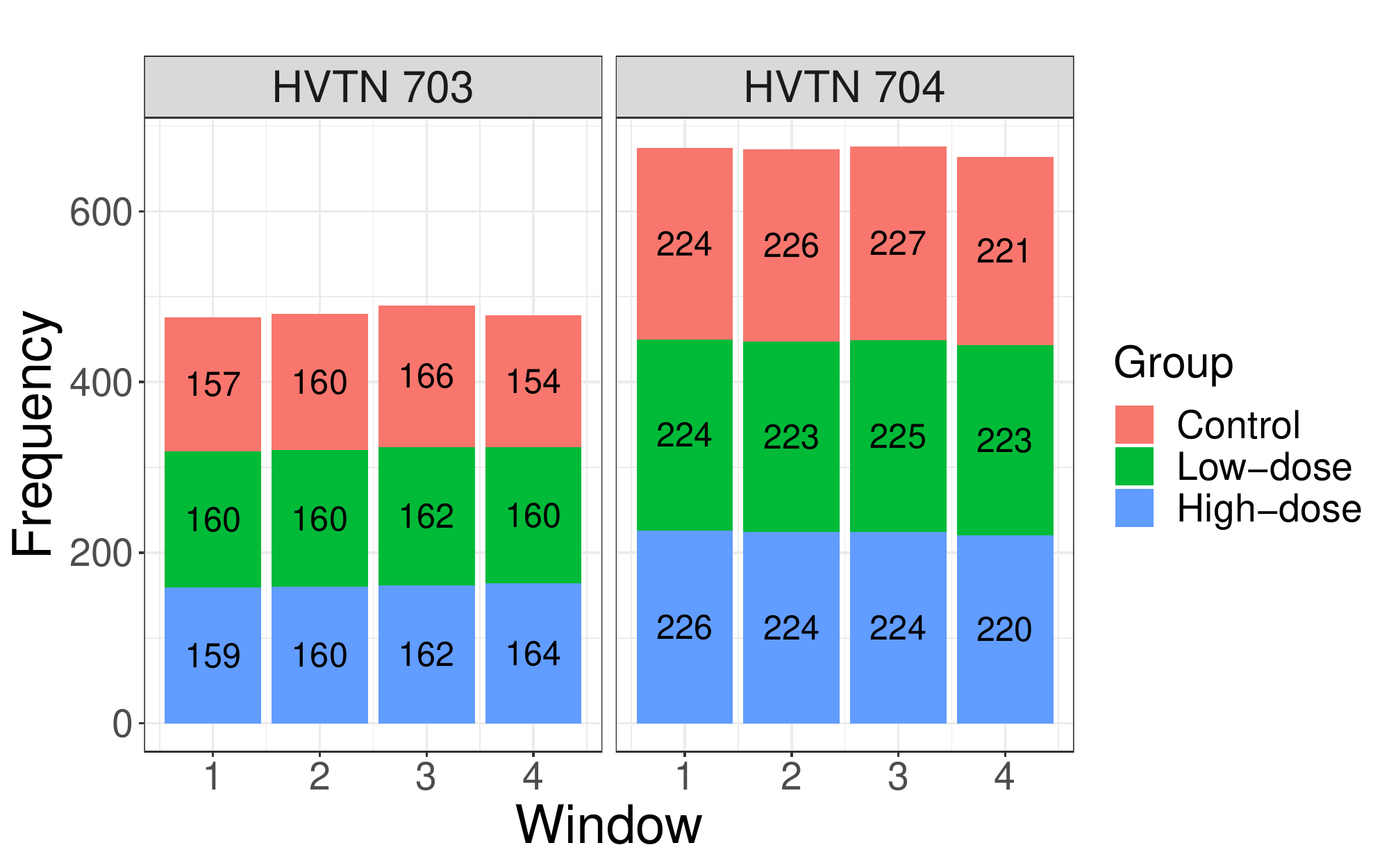}
 \caption{The numbers of participants enrolled for various groups in different windows across the two trials.}
 \label{fig:sample_dist_arms_windows_protocols}
\end{center}
\end{figure}


\newpage
\vskip 2in
\renewcommand\bibname{}
\vspace{-1in}
\bibliographystyle{abbrvnat}

\bibliography{Ref_library}

\begin{thebibliography}{30}
\providecommand{\natexlab}[1]{#1}
\providecommand{\url}[1]{\texttt{#1}}
\expandafter\ifx\csname urlstyle\endcsname\relax
  \providecommand{\doi}[1]{doi: #1}\else
  \providecommand{\doi}{doi: \begingroup \urlstyle{rm}\Url}\fi

\bibitem[Corey et~al.(2021)Corey, Gilbert, Juraska, Montefiori, Morris, Karuna,
  Edupuganti, Mgodi, deCamp, Rudnicki, et~al.]{corey2021two}
L.~Corey, P.~Gilbert, M.~Juraska, D.~Montefiori, L.~Morris, S.~Karuna,
  S.~Edupuganti, N.~Mgodi, A.~deCamp, E.~Rudnicki, et~al.
\newblock Two randomized trials of neutralizing antibodies to prevent {HIV-1}
  acquisition.
\newblock \emph{N Engl J Med}, 384\penalty0 (11):\penalty0 1003--1014, 2021.

\bibitem[Dabrowska(1989)]{Dabrowska1989}
D.~Dabrowska.
\newblock Uniform consistency of the kernel conditional {Kaplan-Meier}
  estimate.
\newblock \emph{Ann Stat}, 17\penalty0 (3):\penalty0 1157--1167, 1989.

\bibitem[D'Agostino et~al.(2003)D'Agostino, Massaro, and
  Sullivan]{DAgostino2003}
R.~D'Agostino, J.~Massaro, and L.~Sullivan.
\newblock Non-inferiority trials: design concepts and issues --- {T}he
  encounters of academic consultants in statistics.
\newblock \emph{Stat Med}, 22\penalty0 (2):\penalty0 169--186, 2003.

\bibitem[Dean et~al.(2020)Dean, Gsell, Brookmeyer, Crawford, Donnelly,
  Ellenberg, Fleming, Halloran, Horby, Jaki, et~al.]{Dean2020}
N.~Dean, P.~Gsell, R.~Brookmeyer, F.~Crawford, C.~Donnelly, S.~Ellenberg,
  T.~Fleming, M.~Halloran, P.~Horby, T.~Jaki, et~al.
\newblock Creating a framework for conducting randomized clinical trials during
  disease outbreaks.
\newblock \emph{N Engl J Med}, 382\penalty0 (14):\penalty0 1366--1369, 2020.

\bibitem[Everson-Stewart and Emerson(2010)]{everson2010bio}
S.~Everson-Stewart and S.~Emerson.
\newblock Bio-creep in non-inferiority clinical trials.
\newblock \emph{Stat Med}, 29\penalty0 (27):\penalty0 2769--2780, 2010.

\bibitem[Fleming(2008)]{Fleming2008}
T.~Fleming.
\newblock Current issues in non-inferiority trials.
\newblock \emph{Stat Med}, 27:\penalty0 317--332, 2008.

\bibitem[Fleming et~al.(2021)Fleming, Krause, Nason, Longini, and
  Henao-Restrepo]{Fleming2021}
T.~Fleming, P.~Krause, M.~Nason, I.~Longini, and A.~Henao-Restrepo.
\newblock {COVID-19} vaccine trials: {The} use of active controls and
  non-inferiority studies.
\newblock \emph{Clin Trials}, 18\penalty0 (3):\penalty0 335--342, 2021.

\bibitem[Follmann et~al.(2021)Follmann, Fintzi, Fay, Janes, Baden, {El Sahly},
  Fleming, Mehrotra, Carpp, Juraska, Benkeser, Donnell, Fong, Han, Hirsch,
  Huang, Huang, Hyrien, Luedtke, Carone, Nason, Vandebosch, Zhou, Cho, Gabriel,
  Kublin, Cohen, Corey, Gilbert, and Neuzil]{Follmann2021}
D.~Follmann, J.~Fintzi, M.~Fay, H.~Janes, L.~Baden, H.~{El Sahly}, T.~Fleming,
  D.~Mehrotra, L.~Carpp, M.~Juraska, D.~Benkeser, D.~Donnell, Y.~Fong, S.~Han,
  I.~Hirsch, Y.~Huang, Y.~Huang, O.~Hyrien, A.~Luedtke, M.~Carone, M.~Nason,
  A.~Vandebosch, H.~Zhou, I.~Cho, E.~Gabriel, J.~Kublin, M.~Cohen, L.~Corey,
  P.~Gilbert, and K.~Neuzil.
\newblock A deferred-vaccination design to assess durability of {COVID-19}
  vaccine effect after the placebo group is vaccinated.
\newblock \emph{Ann Intern Med}, 174\penalty0 (8):\penalty0 1118--1125, 2021.

\bibitem[Freidlin et~al.(2008)Freidlin, Korn, Gray, and
  Martin]{freidlin2008multi}
B.~Freidlin, E.~Korn, R.~Gray, and A.~Martin.
\newblock Multi-arm clinical trials of new agents: {S}ome design
  considerations.
\newblock \emph{Clin Cancer Res}, 14\penalty0 (14):\penalty0 4368--4371, 2008.

\bibitem[Hern\'{a}n and Robins(2020)]{hernan2020causal}
M.~Hern\'{a}n and J.~Robins.
\newblock \emph{{Causal Inference: What If}}.
\newblock Chapman \& Hall/CRC, 2020.

\bibitem[Hobbs et~al.(2018)Hobbs, Chen, and Lee]{Hobbs2018}
B.~Hobbs, N.~Chen, and J.~Lee.
\newblock Controlled multi-arm platform design using predictive probability.
\newblock \emph{Stat Methods Med Res}, 27\penalty0 (1):\penalty0 65--78, 2018.

\bibitem[Howard et~al.(2018)Howard, Brown, Todd, and
  Gregory]{howard2018recommendations}
D.~Howard, J.~Brown, S.~Todd, and W.~Gregory.
\newblock Recommendations on multiple testing adjustment in multi-arm trials
  with a shared control group.
\newblock \emph{Stat Methods Med Res}, 27\penalty0 (5):\penalty0 1513--1530,
  2018.

\bibitem[Howard et~al.(2021)Howard, Hockaday, Brown, Gregory, Todd, Munir,
  Oughton, Dimbleby, and Hillmen]{howard2021platform}
D.~Howard, A.~Hockaday, J.~Brown, W.~Gregory, S.~Todd, T.~Munir, J.~Oughton,
  C.~Dimbleby, and P.~Hillmen.
\newblock A platform trial in practice: {A}dding a new experimental research
  arm to the ongoing confirmatory {FLAIR} trial in chronic lymphocytic
  leukaemia.
\newblock \emph{Trials}, 22\penalty0 (1):\penalty0 1--13, 2021.

\bibitem[Kaizer et~al.(2018)Kaizer, Hobbs, and Koopmeiners]{Kaizer2018}
A.~Kaizer, B.~Hobbs, and J.~Koopmeiners.
\newblock A multi-source adaptive platform design for testing sequential
  combinatorial therapeutic strategies.
\newblock \emph{Biometrics}, 74\penalty0 (3):\penalty0 1082--1094, 2018.

\bibitem[Kopp-Schneider et~al.(2020)Kopp-Schneider, Calderazzo, and
  Wiesenfarth]{Kopp-Schneider2020}
A.~Kopp-Schneider, S.~Calderazzo, and M.~Wiesenfarth.
\newblock Power gains by using external information in clinical trials are
  typically not possible when requiring strict type {I} error control.
\newblock \emph{Biom J}, 62:\penalty0 361--374, 2020.

\bibitem[Krause et~al.(2020)Krause, Fleming, Longini, Henao-Restrepo, Peto,
  Dean, Halloran, Huang, Gilbert, DeGruttola, et~al.]{Krause2020}
P.~Krause, T.~Fleming, I.~Longini, A.~Henao-Restrepo, R.~Peto, N.~Dean,
  B.~Halloran, Y.~Huang, P.~Gilbert, V.~DeGruttola, et~al.
\newblock {COVID-19} vaccine trials should seek worthwhile efficacy.
\newblock \emph{Lancet}, 396\penalty0 (10253):\penalty0 741--743, 2020.

\bibitem[Lee and Wason(2020)]{Lee2020}
K.~Lee and J.~Wason.
\newblock Including non-concurrent control patients in the analysis of platform
  trials: {I}s it worth it?
\newblock \emph{BMC Medical Res Methodol}, 20:\penalty0 165, 2020.

\bibitem[Lee et~al.(2021)Lee, Brown, Jaki, Stallard, and Wason]{Lee2021}
K.~Lee, L.~Brown, T.~Jaki, N.~Stallard, and J.~Wason.
\newblock Statistical consideration when adding new arms to ongoing clinical
  trials: {T}he potentials and the caveats.
\newblock \emph{Trials}, 22\penalty0 (1):\penalty0 1--10, 2021.

\bibitem[Mauri and D'Agostino(2017)]{Mauri2017}
L.~Mauri and R.~D'Agostino.
\newblock Challenges in the design and interpretation of noninferiority trials.
\newblock \emph{N Engl J Med}, 377\penalty0 (14):\penalty0 1357--1367, 2017.

\bibitem[May et~al.(2020)May, Brown, Schmicker, Emerson, Nkwopara, and
  Ginsburg]{May2020}
S.~May, S.~Brown, R.~Schmicker, S.~Emerson, E.~Nkwopara, and A.~Ginsburg.
\newblock Non-inferiority designs comparing placebo to a proven therapy for
  childhood pneumonia in low-resource settings.
\newblock \emph{Clin Trials}, 17\penalty0 (2):\penalty0 129--137, 2020.

\bibitem[Saville and Berry(2016)]{Saville2016}
B.~Saville and S.~Berry.
\newblock Efficiencies of platform clinical trials: {A} vision of the future.
\newblock \emph{Clin Trials}, 13\penalty0 (3):\penalty0 358--366, 2016.

\bibitem[Sridhara et~al.(2015)Sridhara, He, Nie, Shen, and Tang]{Sridhara2015}
R.~Sridhara, K.~He, L.~Nie, Y.~Shen, and S.~Tang.
\newblock Current statistical challenges in oncology clinical trials in the era
  of targeted therapy.
\newblock \emph{Stat Biopharm Res}, 7\penalty0 (4):\penalty0 348--356, 2015.

\bibitem[Tsiatis and Davidian(2021)]{Tsiatis2021}
A.~Tsiatis and M.~Davidian.
\newblock Estimating vaccine efficacy over time after a randomized study is
  unblinded.
\newblock \emph{Biometrics {\rm (in press)}}, 2021.

\bibitem[van~der Laan and Rose(2011)]{van2011targeted}
M.~van~der Laan and S.~Rose.
\newblock \emph{{Targeted Learning: Causal Inference for Observational and
  Experimental Data}}.
\newblock Springer Science \& Business Media, 2011.

\bibitem[van~der Vaart(2000)]{Vaart2000}
A.~van~der Vaart.
\newblock \emph{{Asymptotic Statistics}}.
\newblock Cambridge University Press, 2000.

\bibitem[{WHO Solidarity Trial
  Consortium}(2021)]{WHOSolidarityTrialConsortium2021}
{WHO Solidarity Trial Consortium}.
\newblock Repurposed antiviral drugs for {COVID-19} — {I}nterim {WHO
  Solidarity Trial} results.
\newblock \emph{N Engl J Med}, 384\penalty0 (6):\penalty0 497--511, 2021.

\bibitem[Woodcock and LaVange(2017)]{Woodcock2017}
J.~Woodcock and L.~LaVange.
\newblock Master protocols to study multiple therapies, multiple diseases, or
  both.
\newblock \emph{N Engl J Med}, 377\penalty0 (1):\penalty0 62--70, 2017.

\bibitem[{World Health Organization}(2020)]{WHOR&DBlueprint2020}
{World Health Organization}.
\newblock An international randomised trial of candidate vaccines against
  {COVID-19}.
\newblock Technical Report 0419, World Health Organization, 2020.

\bibitem[Ypma et~al.(2020)Ypma, Johnson, Borchers, Eddelbuettel, Ripley,
  Hornik, Chiquet, and Adler]{NLoptR2020}
J.~Ypma, S.~Johnson, H.~Borchers, D.~Eddelbuettel, B.~Ripley, K.~Hornik,
  J.~Chiquet, and A.~Adler.
\newblock {nloptr:} the {NLopt} nonlinear-optimization package.
\newblock 2020.
\newblock URL \url{https://CRAN.R-project.org/package=nloptr}.
\newblock R package version 1.2.2.2.

\bibitem[Zhang et~al.(2019)Zhang, Chen, Jin, Li, and Quan]{Zhang2019}
L.~Zhang, D.~Chen, H.~Jin, G.~Li, and H.~Quan.
\newblock \emph{{Contemporary Biostatistics with Biopharmaceutical
  Applications}}.
\newblock Springer, 2019.

\end{thebibliography}

\end{document}